%% file: main.tex
\documentclass[lettersize,journal]{IEEEtran}
\IEEEoverridecommandlockouts

\usepackage[numbers]{natbib}

\usepackage{graphicx}
\usepackage[figurename=Figure]{caption}
\usepackage{subcaption}
\usepackage{etoolbox}
\usepackage{scalerel}
\usepackage{hyperref}
\newcommand*{\img}[1]{%
    \raisebox{-.2\baselineskip}{%
        \includegraphics[
        height=\baselineskip,
        width=\baselineskip,
        keepaspectratio,
        ]{#1}%
    }%
}
\usepackage{siunitx}
\usepackage{pifont}
\usepackage{graphicx}
\usepackage{amsmath,amssymb,wasysym,stix}
\usepackage{textcomp}
\usepackage{algorithm}
\usepackage{algcompatible}
\usepackage{algpseudocode}

\usepackage{cancel}
\usepackage{booktabs}
\usepackage{url}
\usepackage{paralist}
\usepackage{color}
\usepackage[usenames, dvipsnames]{xcolor}
\definecolor{FGreen}{RGB}{6,253,105}
\usepackage{epsfig} % for postscript graphics files
\usepackage{bm} % assumes amsmath package installed
\usepackage{listings,lstautogobble}

\usepackage{xspace}
\usepackage{soul,xcolor}
\usepackage{verbatim}
\usepackage{epstopdf}
\usepackage{listings}
\usepackage{parcolumns}
\usepackage{color}
\usepackage{mathtools}

\usepackage{xspace}
\usepackage{soul,xcolor}
\usepackage{float}
\usepackage{verbatim}
\usepackage[utf8]{inputenc}

\usepackage{rotating,multirow}
\usepackage{authblk}

\usepackage[most]{tcolorbox}

\usepackage{amsthm}
\newtheorem{definition}{Definition}

\newtheorem{theorem}{Theorem}
\newtheorem{lemma}{Lemma} 
\newtheorem{problem}{Problem} 
\newtheorem{example}{Example}
  \renewcommand{\footnotesize}{\fontsize{8pt}{11pt}\selectfont}

\DeclareMathAlphabet{\mathcal}{OMS}{cmsy}{m}{n}

\DeclareMathOperator*{\argminB}{argmin}   % Jan Hlavacek

\newcommand{\zono}[1]{\langle #1 \rangle}

\newcommand{\trace}{\text{Tr}}

\usepackage{enumitem}
\begin{document}

\title{\vspace{.85in}
Privacy-Preserving Set-Based Estimation\\ Using Differential Privacy and Zonotopes}

\author{Mohammed M. Dawoud$^{1}$, Changxin Liu$^{2}$, Karl H. Johansson$^{2}$, and Amr Alanwar$^{1,3}$
\thanks{This paper has received funding from the European Union's Horizon 2020 research and innovation program under grant agreement No. 830927.}
\thanks{$^{1}$School of Computer Science and Engineering, Constructor University, Bremen, Germany {\tt\small\{mdawoud, aalanwar\}@constructor.university}}
\thanks{$^{2}$School of Electrical Engineering and Computer Science, KTH Royal Institute of Technology.  {\tt\small \{changxin, kallej\}@kth.se}}
\thanks{$^{3}$School of Computation, Information and Technology, Technical University of Munich. {\tt\small alanwar@tum.de}.}
}

\vspace{1in}
\maketitle
\input{Sections/1-abs.tex}


\input{Sections/2-intro}
\input{Sections/3-prem}

\input{Sections/4-alg}

\input{Sections/5-eval}
\input{Sections/6-con}

\small
\bibliographystyle{ieeetr}
\bibliography{references}

\end{document}

%% file: Sections/1-abs.tex
\begin{abstract}

For large-scale cyber-physical systems, the collaboration of spatially distributed sensors is often needed to perform the state estimation process. Privacy concerns arise from disclosing sensitive measurements to a cloud estimator. To solve this issue, we propose a differentially private set-based estimation protocol that guarantees true state containment in the estimated set and differential privacy for the sensitive measurements throughout the set-based state estimation process within the central and local differential privacy models. Zonotopes are employed in the proposed differentially private set-based estimator, offering computational advantages in set operations. We consider a plant of a non-linear discrete-time dynamical system with bounded modeling uncertainties, sensors that provide sensitive measurements with bounded measurement uncertainties, and a cloud estimator that predicts the system's state. The privacy-preserving noise 
% applied to zonotopes 
perturbs the centers of measurement zonotopes, thereby concealing the precise position of these zonotopes, i.e., ensuring privacy preservation for the sets containing sensitive measurements.
% preserves the privacy of sets representing these sensitive measurements. 
Compared to existing research, our approach achieves less
privacy loss and utility loss through the central and local differential privacy models by leveraging a numerically optimized truncated noise distribution. The proposed estimator is perturbed by weaker noise than the analytical approaches in the literature to guarantee the same level of privacy, therefore improving the estimation utility. Numerical and comparison experiments with truncated Laplace noise are presented to support our approach. 
% Zonotopes, a less conservative form of set representation, 
% represent estimation sets,
% give set operations a computational advantage. The privacy-preserving noise anonymizes the centers of these estimated zonotopes, concealing the precise positions of the estimated zonotopes.

% preserving the privacy of the sets containing the sensitive measurements.

\end{abstract}
% \begin{IEEEkeywords}

%% file: Sections/2-intro.tex
\section{Introduction} \label{sec:intro}
% The widespread use of automated and intelligent systems and the Internet of Things (IoT) in our daily lives has become an apparent fact. One can see it in smart homes, smart devices, intelligent transportation systems (ITS), and intelligent weather forecasting systems. These large-scale systems often need their participating entities to share some data, e.g., vehicle location or pedestrian destination, medical measurements, financial information, and users' habits \cite{8686209}. Sharing these data within the system can cause privacy vulnerabilities, e.g., individuals' positions may be inferred from location data shared with routing applications such as Waze and Google Maps \cite{10.1145/2508859.2516735,6256763,proloc}. 

\IEEEPARstart{T}{he} ubiquity of automated and intelligent systems, coupled with the proliferation of the Internet of Things (IoT), has become an undeniable reality in our daily lives. From smart homes and devices to Intelligent Transportation Systems (ITS) and advanced weather forecasting, these large-scale systems rely heavily on the exchange of data among their constituents. However, this data sharing presents significant privacy concerns, as sensitive information such as vehicle locations, medical measurements, financial data, and user behaviors are often shared within the system \cite{8686209}.

% Another example of a privacy vulnerability is identity disclosure for a driver when he shares his real identity on vehicular ad-hoc networks (VANETs) of ITS. Thus, there is a pressing need for privacy-preserving mechanisms for data publishing within such systems. Motivated by this, this work is dedicated to privacy-preserving state estimation.

For instance, sharing location data with navigation applications like Waze and Google Maps can inadvertently expose individuals' whereabouts, compromising their privacy \cite{10.1145/2508859.2516735,6256763,proloc}. Similarly, in Vehicular Ad-hoc Networks (VANETs) within ITS, sharing real identities can lead to the disclosure of drivers' identities, posing a risk to their privacy. Consequently, there is an urgent need for privacy-preserving mechanisms to safeguard data sharing within such systems. Additionally, 
% safety-critical applications necessitate guaranteed state inclusion within bounded sets to ensure the avoidance of unsafe conditions.
when input disturbances and observation errors are unknown but bounded, safety-critical applications demand guaranteed state inclusion within bounded sets to prevent unsafe conditions.
% Motivated by this, this work is dedicated to privacy-preserving state estimation.
Motivated by these challenges, this work presents a privacy-preserving mechanism tailored for sets containing sensitive measurements. This mechanism aims to safeguard the privacy of these sensitive measurements throughout the set-based state estimation process.
% during the set-based state estimation process.
\subsection{Set-Based Estimation}

Estimation is a ubiquitous phenomenon that plays a crucial role in decision-making across various domains. From an estimated value-based perspective, state estimation can be categorized into point-wise and set-based estimation. The former involves using measurements to calculate a single point-valued estimate, such as determining the location of an autonomous vehicle and the distance between space objects. On the other hand, set-based estimation utilizes measurements along with the system model, if available, to predict a set that encompasses the system's state. With advancements in sensors' capabilities, set-based estimation has become increasingly prevalent in applications such as robotics and autonomous vehicles \cite{conf:setloc}.
% Supported by the development of sensors' capabilities, the widespread of autonomous systems, which massively deploy set-based estimation, can be found in applications including robots and autonomous vehicles \cite{doi:10.1177/1729881419839596,s21062085,vandana}.
% The set-based estimation uses, in general, different set representations, such as intervals, polytopes, ellipsoids, zonotopes, and constrained zonotopes \cite{ALTHOFF2010233, Kurzhanskiy200726, ALAMO20051035, SCOTT2016126}. This paper considers zonotopes for set representation as they are computationally efficient and are closed under multiple set operations \cite{9838494}.
The set-based estimation typically employs various set representations, including intervals, polytopes, ellipsoids, zonotopes, and constrained zonotopes \cite{ALTHOFF2010233, Kurzhanskiy200726, SCOTT2016126}. In this paper, we choose zonotopes for set representation due to their computational efficiency and closure under multiple set operations \cite{10.1007/978-3-540-31954-2_19}.

% There are three approaches to set-based state estimation: the model-based approach, the data-driven approach, and the hybrid approach.
Model-based and data-driven approaches are two distinct methods for set-based state estimation. The model-based approach relies on the system model and sensors' measurements to estimate the system's states \cite{conf:disdiff}.
% Some of the filters used in the model-based approach can estimate states for systems with a linear model and a Gaussian noise distribution, such as the Kalman filter \cite{7549014}.
% Filters such as the Kalman filter are commonly utilized within this paradigm, particularly suited for systems characterized by linear models and Gaussian noise distributions \cite{7549014}.
% On the other hand, as systems become more automated, the number of sensors in the system and the amount of data generated per sensor increase.
However, as systems evolve towards greater automation, the proliferation of sensors and the corresponding increase in data volume pose challenges. Developing system models for such complex systems may be prohibitively expensive \cite{doi:10.1177/0142331219879858}, making the data-driven approach more practical in these scenarios \cite{9838494}.
% . Therefore, a data-driven approach might be more convenient for such complex systems \cite{https://doi.org/10.1002/qj.3551,DBLP:journals/corr/abs-1806-03753}.
% The hybrid estimation approach combines the advantages of both methods. It leads to fast state calculation as in the model-based approach and handles system complexity as in the data-driven approach \cite{JIN2020105962}. This paper considers the model-based approach.
% The hybrid estimation approach merges the strengths of both the model-based and data-driven approaches. By combining the rapid state calculation capabilities of the model-based approach with the ability to handle system complexity inherent in the data-driven approach, the hybrid approach offers a promising compromise. This paper considers the model-based approach.

\subsection{State Estimation Under Privacy Constraints}
Existing private estimation strategies can be categorized into encryption-based methods and non-encryption-based methods. For the former, the authors in \cite{kim2019encrypted} developed a private Luenberger observer based on additively homomorphic encryption, where the state matrix only consists of integers such that the observer operates for an infinite time horizon. Also, a partially homomorphic encryption scheme was utilized in \cite{proloc,zhang2020secure,emad2022privacy} to develop a private state estimator. Most recent works in the latter are based on Differential Privacy (DP), a powerful non-cryptographic technique for quantifying and preserving individuals' privacy. 
In essence, DP characterizes a property of a randomized algorithm, ensuring that its output remains stable regardless of any changes made to an individual's information in the database. This stability protects individuals' privacy by mitigating privacy-related attacks.
% Roughly, DP refers to a property of a randomized algorithm; its output remains stable for any changes that occur to an individual in the database, therefore protecting against privacy-related attacks on the individual's information.

The DP encompasses two models: the central differential privacy (CDP) model, where a trusted data curator executes the DP algorithm on a centrally held dataset, and the local differential privacy (LDP) model, where each participant independently executes the DP algorithm \cite{9253545}. However, the accuracy of this output is negatively affected by the randomization of the DP algorithms. In \cite{6606817}, a general framework was developed for a differentially private filter. The authors in \cite{degue2017differentially} proposed a differentially private Kalman filter. The methodology has also been generalized to non-linear systems by the authors in \cite{https://doi.org/10.1002/rnc.439}. Note that the above DP observers are point-wise ones, and disturbances or uncertain parameters, which are only known to be bounded, are present for many systems. Hence, set-based estimators and interval estimators can handle state estimation for such systems. Recently, the authors in \cite{9147726} developed a differentially private interval observer with bounded input perturbation. Inspired by the work in \cite{9147726,ALANWAR2023100786}, our work in \cite{10178269} utilizes a truncated additive mechanism with a numerically optimized noise distribution within the context of the CDP model and employs zonotopes for set representation to present a differentially private set-based estimator for linear discrete-time systems. 
\subsection{Contributions}
This paper presents a differentially private set-based estimator that ensures true state containment within the estimated set and provides differential privacy for sensitive measurements throughout the set-based state estimation process. It extends the work in \cite{10178269} to non-linear discrete-time systems within the context of the CDP and LDP models and evaluates this extension on real-world data. 
% Particularly, this paper employs zonotopes for set representation and utilizes a truncated additive mechanism with a numerically optimized noise distribution within the context of the CDP and LDP models to present an $(\epsilon,\delta)$-ADP set-based estimator for non-linear discrete-time systems with bounded modeling and measurement uncertainties. The proposed estimator ensures the privacy of sets containing the sensitive sensors' measurements throughout the estimation process with minimal utility loss.
% The provided comparative results illustrate the improvement in the estimation utility.
Specifically, we continue to utilize zonotopes for set representation and employ a truncated additive mechanism with a numerically optimized noise distribution in both models to present a differentially private set-based estimator for non-linear discrete-time systems with bounded modeling and measurement uncertainties. The proposed estimator ensures the privacy of sets containing sensitive sensors' measurements throughout the estimation process with minimal utility loss. Comparative results provided in this paper demonstrate the enhancement in estimation utility.

The main contributions of this article can be summarized as follows:
 \begin{itemize}

\item We introduce a differentially private set-based estimator leveraging a truncated additive mechanism with a numerically optimized noise distribution \cite{DBLP:journals/corr/abs-2107-12957} and employing zonotopes for set representation. The proposed estimator ensures the privacy of sets containing sensors' measurements throughout the estimation process with minimal utility loss.

\item We apply the privacy-preserving noise to the zonotopes containing the sensitive measurements within the context of the CDP and LDP models, concealing the precise positions of these zonotopes.
% The precise positions of the estimated zonotopes are concealed. Furthermore, the computation of the estimated zonotopes' centers and the private measurements is inherently decoupled from the computation of the generators. 
% {\color{blue}\item We employ the Central Differential Privacy (CDP) and Local Differential Privacy (LDP) models to introduce privacy-preserving noise to the sensors’ measurements.}
%We apply the privacy-preserving noise to the sensors' measurements using the CDP and the LDP models.}
\item Through our comprehensive evaluation of the proposed differentially private set-based estimator, we illustrate that the truncated Laplace distribution \cite{9147726} necessitates a more significant amount of noise compared to the truncated optimal noise distribution to achieve an equivalent level of privacy under both the CDP and LDP models.
% We evaluate the proposed $(\epsilon,\delta)$-ADP set-based estimator by comparing the deployment of a truncated optimal noise distribution with a truncated Laplace distribution, as detailed in \cite{9147726}. Our evaluation demonstrates that the latter requires a more substantial noise than the former to achieve an equivalent privacy level using the CDP and the LDP models.
%and show that the latter leads to a more substantial noise than the former to achieve the same privacy level. Consequently, 
%the bounds of the optimal truncated noise are less than those of
% As a result, the proposed $(\epsilon,\delta)$-ADP set-based estimator exhibits enhanced utility, as measured by the average estimation error, in comparison to the existing method.
%achieves improved utility, represented as the average estimation error, over the existing method. All used data and code are publicly available
\end{itemize}
All utilized data and code are publicly available\footnote {\label{footnote:DP-estimator}\url{https://github.com/mohammed-dawoud/Differentially-Private-Set-Based-Estimation-Using-Zonotopes}}.

The rest of the paper is organized as follows: The preliminaries and problem statement are presented in Section~\ref{sec:preminilaries}. The algorithms are designed and evaluated in Section~\ref{sec:main} and Section~\ref{sec:eval}, respectively. Finally, we conclude the work in Section~\ref{sec:conc}.

%% file: Sections/3-prem.tex
\section{Preliminaries and Problem Setup}
\label{sec:preminilaries}
In this section, we introduce some notation,
 % to be used throughout this work and
 present the needed preliminaries, and then formulate the problem.
\subsection{Notation} Throughout this paper, we denote vectors and scalars by lower case letters, matrices by upper case letters, the set of real numbers by $\mathbb{R}$, the set of integers by $\mathbb{Z}$, and the set of positive integers by $\mathbb{Z}^{+}$. For a given vector $v$, we denote the matrix with $v$ on the diagonal as diag$(v)$. For a given matrix $M\in \mathbb{R}^{m\times n}$, its transpose is given by $M^T$ and its Frobenius norm is given by $\lVert M \rVert_{F} = \sqrt{\trace{(M^T M)}}$. We denote the $i$-th element of a vector or list $a$ by $a^{(i)}$. For a vector $a\in \mathbb{R}^n$, we denote its $L_2$ norm by $| a |_2 = \sqrt{\sum^n_{i=1}(a^{(i)})^2}$.
% \st{and $L_\infty$ norm by $\lVert y \rVert_{\infty} =\max_{i\in\{1,\dots\,n\}}{y^{(i)}}$}
 For a scalar or a vector $a_k$, we use $\{a_k\}$ to denote a list of multiple $a_k$ such that $k \in \{1,\dots, k_1\}$ and $k_1\in \mathbb{Z}^{+}$.
 % We use $\{y_k\}$ to denote a list of the measurement vectors $y_k$ at different time steps $k$.
%k Given $y_{k}\in \mathbb{R}^n$ is a vector of measurements at time step $k$,
 We denote the $L_2$ norm of a vector-valued signal $y$ composed of a list of $y_k \in \mathbb{R}^n$, i.e., $y=\{y_k\}$,
 % $y:\mathbb{Z}^{+} \rightarrow \mathbb{R}^n$,
 by $\lVert y \rVert_2 = \sqrt{\sum^{\infty}_{k=1}(|y_k|_2)^2}$.
\subsection{Set Representation and Set-based Estimation}
Next, we review set representation, set operations, and set-based estimation based on zonotopes.
\subsubsection{Set Representation and Operations}

The zonotope is defined as follows:
\begin{definition} [Zonotope \cite{z-lop-95}]
\label{def:zonotope} 
Given a center $c_{\mathcal{Z}} \in \mathbb{R}^n$ and $\gamma_{\mathcal{Z}} \in \mathbb{N}$ generator vectors in a generator matrix $G_{\mathcal{Z}}=\begin{bmatrix} g_{\mathcal{Z}}^{(1)}& \dots &g_{\mathcal{Z}}^{(\gamma_{\mathcal{Z}})}\end{bmatrix} \in \mathbb{R}^{n \times \gamma_{\mathcal{Z}}}$, a zonotope is defined as
%\cite[Def.~1]{Girard2005} 
\begin{align}\label{eqn:zonotope}
	\mathcal{Z} = \Big\{ x \in \mathbb{R}^n \; \Big| \; x = c_{\mathcal{Z}} + \sum_{i=1}^{\gamma_{\mathcal{Z}}} \beta^{(i)} \, g^{(i)}_{\mathcal{Z}} \, ,
	-1 \leq \beta^{(i)} \leq 1  \Big\}.
\end{align}
We use the shorthand notation $\mathcal{Z} = \zono{c_{\mathcal{Z}},G_{\mathcal{Z}}}$ for such a zonotope. %\hfill $\square$
\end{definition}
Zonotopes are closed under linear maps and Minkowski sum \cite{conf:althoffthesis}. \label{item:linear-map}
%Zonotopes are closed under linear transformation. 
The linear map $L \in \mathbb{R}^{m  \times n}$ for zonotope ${\mathcal{Z}}$ is defined and computed as follows:
\begin{align}
L \mathcal{Z} = \{Lz | z\in\mathcal{Z}\}  = \zono{L c_{\mathcal{Z}}, L G_{\mathcal{Z}} }. \label{eqn:linmap}
\end{align}
%A linear map $L$ is defined as $L \mathcal{Z}  = \zono{L c_{\mathcal{Z}}, L G_{\mathcal{Z}}}$.
%ii) \textit{Minkowski Sum}.
\label{item:minkowski-sum}
%Zonotopes are also closed under Minkowski sum.
Given two zonotopes $\mathcal{Z}_1=\langle c_{\mathcal{Z}_1},G_{\mathcal{Z}_1} \rangle$ and $\mathcal{Z}_2=\langle c_{\mathcal{Z}_2},G_{\mathcal{Z}_2} \rangle$, the Minkowski sum is defined and computed as %can be computed exactly \cite{conf:zono1998}:
\begin{align}\label{eqn:minkowski-sum}
     \mathcal{Z}_1 \oplus \mathcal{Z}_2 &= \Big\{ z_1 + z_2 | z_1 \in \mathcal{Z}_1, z_2 \in \mathcal{Z}_2 \Big\}\nonumber\\&= \Big\langle c_{\mathcal{Z}_1} + c_{\mathcal{Z}_2}, [G_{\mathcal{Z}_1}, G_{\mathcal{Z}_2}]\Big\rangle.
\end{align}

The Cartesian product is 
\begin{align}\label{eqn:cartesian-product}
\mathcal{Z}_1 \times \mathcal{Z}_2 &= \bigg\{ \begin{bmatrix}z_1 \\ z_2\end{bmatrix} \bigg| z_1 \in \mathcal{Z}_1, z_2 \in \mathcal{Z}_2 \bigg\} \nonumber\\
&= \Bigg\langle \begin{bmatrix} c_{\mathcal{Z}_1} \\ c_{\mathcal{Z}_2} \end{bmatrix}, \begin{bmatrix} G_{\mathcal{Z}_1} & 0 \\ 0 & G_{\mathcal{Z}_2}\end{bmatrix} \Bigg\rangle.
\end{align}
\\

\subsubsection{{Set-Based Estimation}}
\label{subsubsection:set-based-estimation}
% \subsubsubsection{{Linear Discrete-Time Model}}
% \label{subsubsection:set-based-estimation-linear-model}
%k In addition to the inherent presence of process and measurement noises in systems, some practical applications (e.g., the localization of a self-driving car) introduce non-linearity into predicted state and measurement behaviors. Therefore, we
Consider the following non-linear discrete-time dynamical system with bounded modeling and measurement uncertainties: %k \cite{9811706}
% Consider a linear discrete-time system with bounded noise given 
\begin{equation}\label{eqn:system-model-non-linear}
\begin{aligned}
    x_{k+1} &= f(x_k) + w_k,\\
    y_{k}^{(i)} &= h^{(i)}(x_k) + v_{k}^{(i)},
\end{aligned}
\end{equation}
where $x_k \in \mathbb{R}^n$ is the system state at time step $k \in \mathbb{Z}^+$, and $y_{k}^{(i)}\in \mathbb{R}$ is the measurement of sensor $i \in \{1,\dots,\;m\}$ with $m$ equals the number of available sensors. 
The functions $f$ and $h^{(i)}$ are assumed to be differentiable.
% {The matrices }$F\in \mathbb{R}^{n\times n}$ and $H^{(i)} \in \mathbb{R}^{1\times n}$ are state and measurement matrices, respectively.
The vector \begin{math}w_k\end{math} and the scalar \begin{math} v_{k}^{(i)}\end{math} are process and measurement noise, respectively. They are assumed to be unknown but bounded by the zonotopes $\mathcal{Z}_w= \zono{0, G_w}$, and  $\mathcal{Z}^{(i)}_{v}= \zono{0, G^{(i)}_{v}}$, respectively (If the noise zonotopes are not centered around zero, the resulting estimates will be shifted). The system has a bounded initial state $x_0\in \mathcal{\Bar{Z}}_0= \zono{\Bar{c}_0,\Bar{G}_0}$.
% and a predicted state set $\mathcal{\hat{Z}}_{k}$.
At each time step $k$, the set-based state estimator aims to find the corrected state set $\mathcal{\Bar{Z}}_{k}$ by finding the intersection between the predicted state set $\mathcal{\hat{Z}}_{k}$ and the measurement sets $\mathcal{Z}^{(i)}_{y_k}=\zono{y^{(i)}_k, G^{(i)}_{v}}$ corresponding to the sensor measurements $y^{(i)}_k$ and measurement uncertainties $\mathcal{Z}^{(i)}_{v}$, $i=1,\dots,m$ \cite{9838494, ALANWAR2023100786,9811706}. The set-based state estimator is described by Algorithm \ref{alg:set-estimation-using-zonotope}

\begin{algorithm}
\caption{Set-Estimation Using Zonotopes}\label{alg:set-estimation-using-zonotope}
\begin{algorithmic} [1]
{\Statex {\bfseries Input:} Process noise zonotope $\mathcal{Z}_w= \zono{0, G_w}$, and 
  measurement noise zonotopes $\mathcal{Z}^{(i)}_{v}= \zono{0, G^{(i)}_{v}}$, where $i \in \{1, \dots, m\}$.
} 
{ \Statex {\bfseries Output:} $\mathcal{\Bar{Z}}_{k}=\zono{\Bar{c}_k, \Bar{G}_k}$.}
\Statex{\bfseries Initialization:} Set $k = 1$ and $\mathcal{\Bar{Z}}_{0} = \zono{\Bar{c_0},\Bar{G_0}}$.
\While{$True$}

\State The cloud estimator uses the $m$ measurement sets $\mathcal{Z}^{(i)}_{y_k}$ containing the sensitive measurements to perform the estimation process according to the following steps:
\begin{itemize}
    \item {Prediction step:} The set-based state estimator determines the predicted state set $\mathcal{\hat{Z}}_{k}$ using a Taylor series expansion as described in \cite{conf:thesisalthoff}. The predicted state set is determined based on the past corrected state set $\mathcal{\Bar{Z}}_{k-1}$ and a process noise zonotope $\mathcal{Z}_w$. This is given by
\begin{equation}\label{eqn:predicted-state-set}
\mathcal{\hat{Z}}_{k}=f(\mathcal{\Bar{Z}}_{k-1})\oplus \mathcal{Z}_w.
\end{equation}\label{line:predicted-state-set}
%in \eqref{eqn:time-update}.
\item {Correction step:} The corrected state set $\mathcal{\Bar{Z}}_{k}$ is determined by the reduction of $\mathcal{\Acute{Z}}_k = \zono{\Acute{c}_k, \Acute{G}_k}$, which over-approximates the intersection between the predicted state set $\mathcal{\hat{Z}}_{k}=\zono{\hat{c}_{k}, \hat{G}_{k}}$ and the $m$ measurement sets $\mathcal{Z}^{(i)}_{y_k}$. This is given by
\begin{align}
  \mathcal{\Acute{Z}}_k \supseteq \mathcal{\hat{Z}}_{k} \cap_{i=1}^m \mathcal{Z}^{(i)}_{y_k}.
  \label{eqn:intersection-over-approximation}
\end{align} 
\item {Update the time:} $k=k+1$. \label{line:time-step-alg-1}
\end{itemize}
\EndWhile
\end{algorithmic}
\end{algorithm}

% The prediction update aims to obtain $\mathcal{\hat{Z}}_{k}=\zono{\hat{c}_{k}, \hat{G}_{k}}$, which is computed according to the system model defined in \eqref{eqn:system-model} by
% \begin{align}\label{eqn:time-update}
%    \mathcal{\Acute{Z}}_{k}=F \mathcal{\Bar{Z}}_{k-1}\oplus \mathcal{Z}_w,
% \end{align}
To obtain the corrected state set zonotope $\mathcal{\Bar{Z}}_{k}=\zono{\Bar{c}_{k}, \Bar{G}_{k}}$ with $\Bar{c}_k=\Acute{c}_k$, the order of the generator matrix $\Acute{G}_k$ of $\mathcal{\Acute{Z}}_{k}$ is reduced as follows:
\begin{align}\label{eqn:girard-reduction}
    \Bar{G}_{k} =\downarrow_q	\Acute{G}_k,
\end{align}
where $\downarrow_q$ denotes the reduction in the order of the generator matrix according to \cite{10.1007/978-3-540-31954-2_19}.

\subsection{{Differential Privacy}}
% The DP aims to provide a certain level of privacy protection for an individual's data.
We will present a few notions on DP, which will be used later to develop a differentially private set-based estimator.  Let $\mathcal{H}$ denote the space of datasets of interest (e.g., sensor measurements) \cite{degue2017differentially}. We define a symmetric binary relation on $\mathcal{H}$, called adjacency and denoted by Adj, in which two datasets $h,\Acute{h}\in \mathcal{H}$ are called adjacent, denoted by Adj$(h,\Acute{h})$, if they differ by the value of exactly one individual's data \cite{TCS-042}. Given a pair of adjacent datasets $h,\Acute{h}\in \mathcal{H}$, 
% i.e., Adj$(h,\Acute{h})$,
a differentially private randomized mechanism ${M}$ with an anonymized measurable output space $\mathcal{O}$, aims to prevent an adversary from inferring knowledge about an individual's data by generating randomized outputs $M(h), M(\Acute{h}) \in \mathcal{O}$ with close distributions for adjacent inputs. {We use a pair of non-negative constants $(\epsilon,\delta)$ to quantify the privacy loss \cite{10.1007_11787006_1,degue2017differentially,9348921}. 
% In particular, $e^{\epsilon}$ represents the ratio of probabilities to obtain the same state estimates with adjacent global measurement signals $y$ and $\Acute{y}$, and $\delta$ allows a small degree of violation when bounding the ratio of probabilities.}

\begin{definition} [Approximate Differential Privacy (ADP) \cite{10.1007_11787006_1,degue2017differentially}]
\label{def:approximate-dp} 
% Let the space of measurements $y$ denoted by $\mathcal{D}$. 
A randomized mechanism ${M}$, which maps
% the datasets' space 
$\mathcal{H}$ equipped with the adjacency relation Adj$(h,\Acute{h})$ to a measurable space $\mathcal{O}$, is $(\epsilon,\delta)$-ADP, $\epsilon,\delta \geq 0$, if $\forall S \in \mathcal{O}$ and $\forall \; h,\Acute{h}\in \mathcal{H}$ such that Adj$(h,\Acute{h})$,
\begin{align}\label{eqn:adp-formula}
	Pr\bigl[M(h)\in S \bigl]\leq e^{\epsilon} Pr\bigl[M(\Acute{h})\in S \bigl]+\delta. 
\end{align}
% such that $\epsilon,\delta \geq 0$. 
If $\delta=0$, ${M}$ is said to be $\epsilon$-differentially private.
%\hfill $\square$
\end{definition}

For the set-based estimation problem, the adjacent datasets $h,\Acute{h}\in \mathcal{H}$ can be the centers of measurement sets.

To motivate the work in this paper, consider the following scenario:
\begin{example}
An intrusion detection system installed over a region consists of light imaging, detection, and ranging (LIDAR) sensors distributed throughout the area and an untrusted cloud estimator. This cloud estimator estimates the set of possible locations of an intruding quadcopter in the region using distance measurements with bounded noise provided by the LIDAR sensors. The cloud estimator utilizes these measurement sets to determine the set of possible locations of the intruding quadcopter. We aim to safeguard the measurements against untrusted parties throughout the estimation process.
\end{example}

\begin{figure}[h]
\graphicspath{ {./Figures/} }
    \centerline{
\includegraphics [scale=0.28]{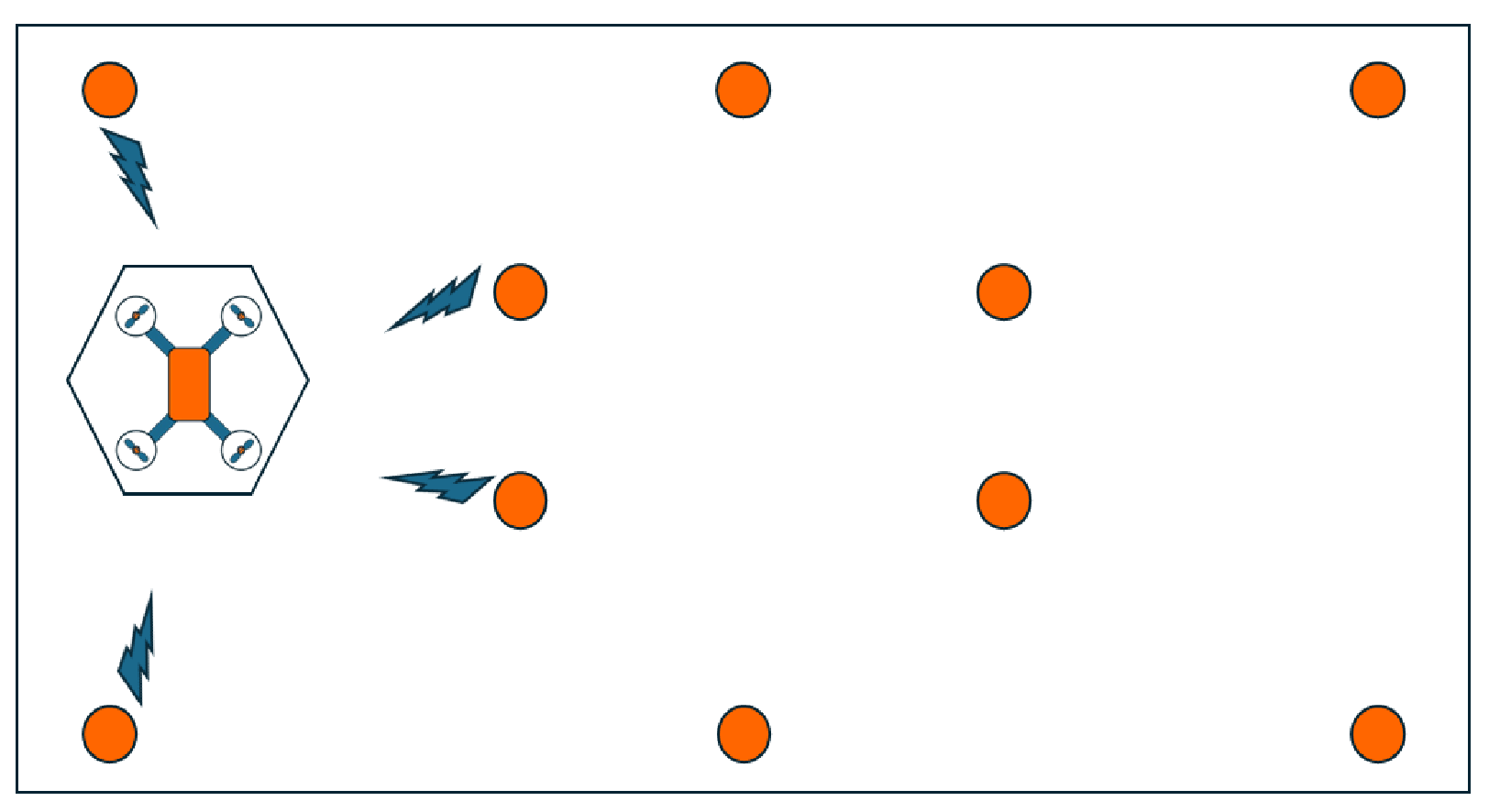}}
    \caption{Intrusion detection system installed over a region,  {\img{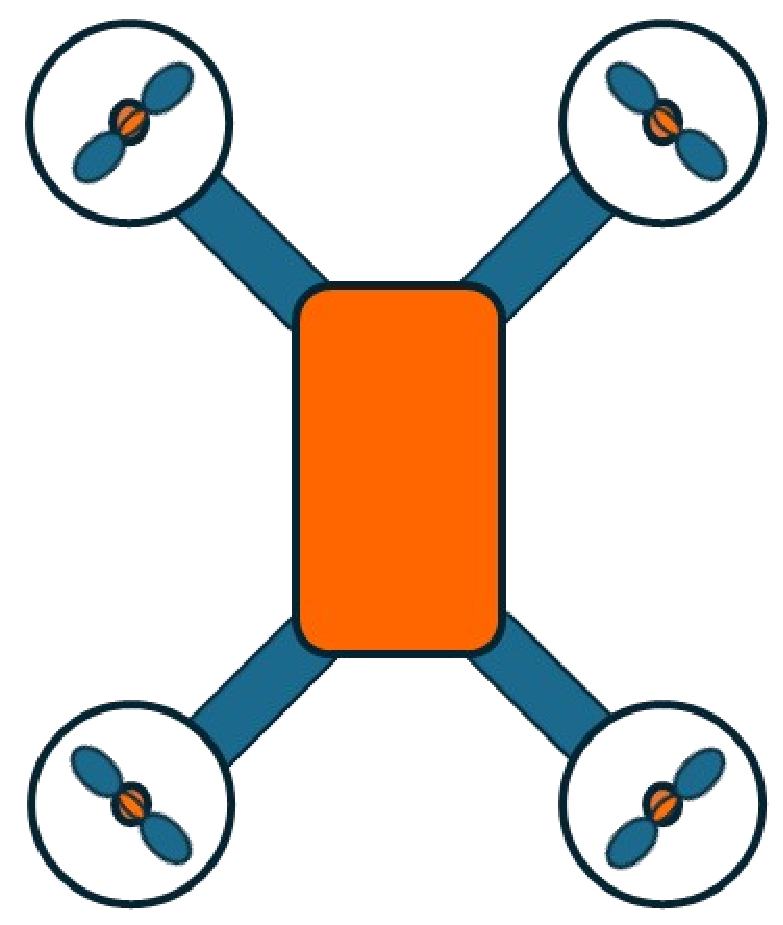}}: the quadcopter, {\img{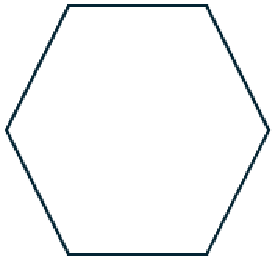}}
    : the estimated set, {\img{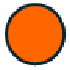}}: LIDAR sensors.}
    \label{fig:example}
\end{figure}

\subsection{Problem Setup}
Consider the following entities for the two cloud-based state estimation setups visualized in Figures \ref{fig:local-DP-Setting} and \ref{fig:Centralized-DP-Setting} \cite{ALANWAR2023100786}.

\begin{itemize}
\item {Plant:} A system that we aim to estimate its set of possible states. We consider a system with a publicly known non-linear discrete-time model described in \eqref{eqn:system-model-non-linear}. 
% Notably, one of these systems follows a linear model, while the other exhibits nonlinearity.
%described by the publicly known linear discrete-time model in \eqref{eqn:system-model} with unknown but bounded disturbances.
\item {Sensors:} An array of $m$ sensors, where each entity $i$ produces a private measurement denoted as $y_k^{(i)}$, with $i \in {1, \dots, m}$.
% An array of $m$ sensors that provides a vector of measurements $y_k=\big[{y_k^{(1)}},\dots,{y_k^{(m)}}\big]^T$, $k\in \mathbb{Z^+}$. Each sensor $i$ produces a private measurement $y^{(i)}_k$, $i\in \{1,\dots, m\}$.
\item {Sensor Manager:} An entity with computational capabilities that enable it to aggregate the measurements of all sensors and perturb them with differential privacy noise. 
\item {Cloud Estimator:} An untrusted entity that performs set-based estimation for the system state.
\end{itemize}

The setup in Figure \ref{fig:local-DP-Setting} uses the LDP model of DP, in which each sensor $i \in \{1,\dots,\;m\}$ locally perturbs its measurement with the privacy-preserving noise before transmitting it to the cloud estimator. The data sets $h,\Acute{h}\in \mathcal{H}$ in this setup are the centers of measurement sets of a single sensor $i$. The setup in Figure \ref{fig:Centralized-DP-Setting} uses the CDP model of DP, in which the sensor manager acts as a trusted data curator, adding privacy-preserving noise to the measurements of $m$ sensors. The data sets in this setup are the centers of measurement sets of $m$ sensors. In both setups, 'protected' refers to measurement sets that have been safeguarded with privacy-preserving noise, while 'unprotected' refers to those without any privacy-preserving noise.

\begin{problem}{We aim to design a differentially private set-based estimator for a plant with the model described in \eqref{eqn:system-model-non-linear}, such that the state estimates are obtained while preserving the privacy of the measurement sets. This estimator obtains state estimates while 
% that are insensitive to the measurement of a single sensor, thereby 
keeping each sensor's measurement private to the sensor and protecting it from any untrusted entity (e.g., the cloud estimator). This is accomplished within the context of the LDP model for the setup in Figure~\ref{fig:local-DP-Setting}  and within the context of the CDP model for the setup in Figure~\ref{fig:Centralized-DP-Setting}.}
\end{problem}
\begin{figure}[H]
    \centering
\begin{subfigure}[H]{0.5\textwidth}
\graphicspath{ {./Figures/} }
    \centerline{
\includegraphics
[width=\textwidth]{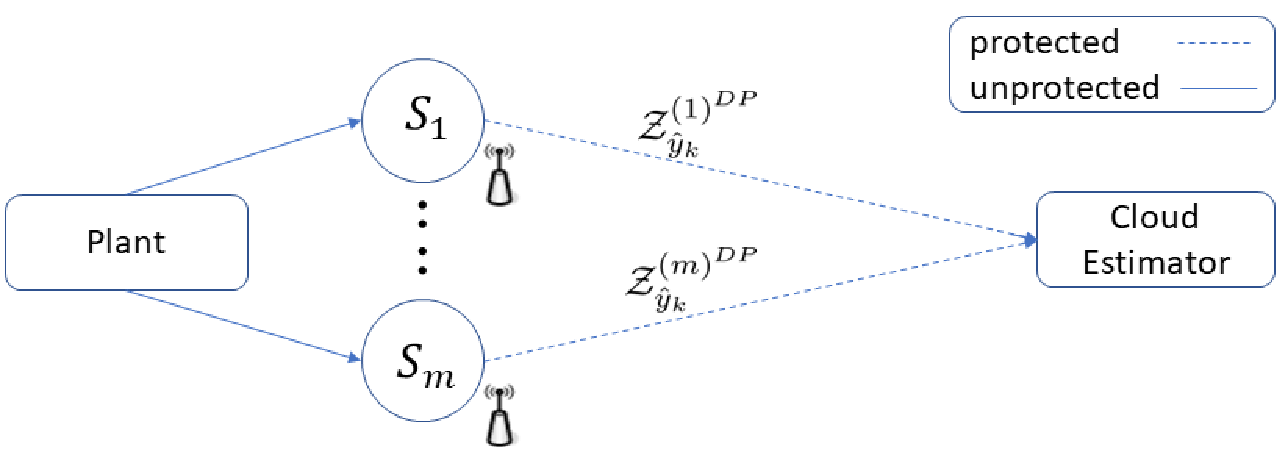}}
    \caption{The cloud estimator is set up within the context of the LDP model.}
    \label{fig:local-DP-Setting}
\end{subfigure}

\begin{subfigure}[H]{0.5\textwidth}
\graphicspath{ {./Figures/} }
    \centerline{
\includegraphics
[width=\textwidth]{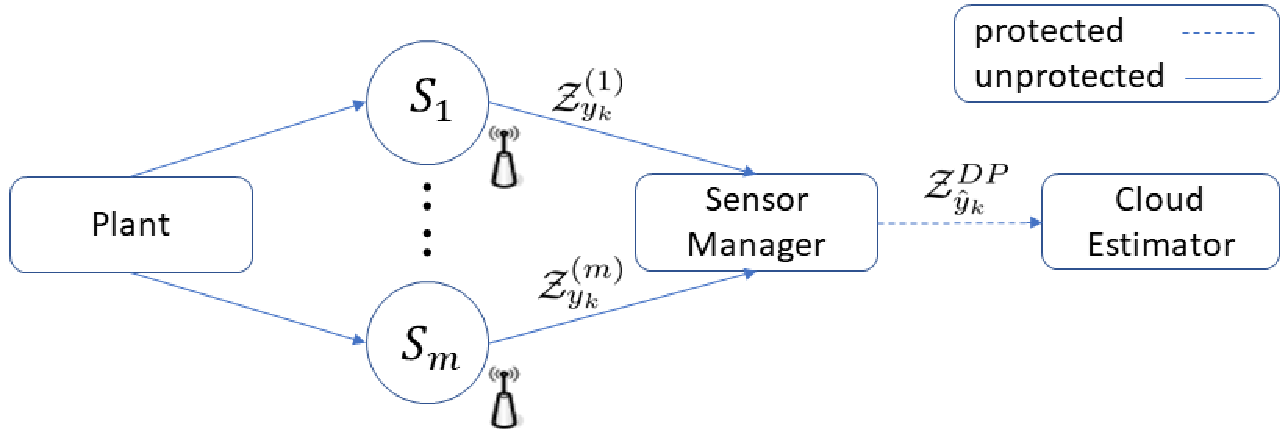}}
    \caption{The cloud estimator is set up within the context of the CDP model.}
    \label{fig:Centralized-DP-Setting}
\end{subfigure}
    \caption{The setups of the cloud estimator within the CDP and LDP models.}
    \label{fig:cloud-estimator-setup}
\end{figure}

%% file: Sections/4-alg.tex
\section{Differentially Private Set-Based Estimation}\label{sec:main}

In this section, we develop the differentially private set-based estimator. Based on the characteristics of the zonotopic representation of sets, the actual positions of these sets in the measurement space can be concealed if we perturb their centers. Thus, we protect the measurement sets by adding privacy-preserving noise to their centers, thereby protecting their positions. In subsection \ref{subsection:Design-of-Truncated-Optimal-Additive}, an additive noise mechanism employs the LDP model, while the other additive noise mechanism utilizes the CDP model to introduce perturbations to the sensitive measurements. 
% Specifically, we utilize the CDP model for one mechanism and the LDP model for the other.
In both cases, we incorporate optimal noise characterized by a numerically generated truncated distribution. These perturbations result in protected measurements. Then, in subsection \ref{subsection:Differentially-Private-Zonotope-based-Set-Membership-Estimation}, we deploy the measurement sets containing the protected measurements to the set-based estimator. % enabling
This enables the derivation of state estimates while keeping the measurements protected against any untrusted party.
% that exhibit insensitivity to a single sensor's measurement.
% The complete differentially private estimator scenario is summarized in Algorithm \ref{alg:adp-set-estimation-using-zonotope}.

\subsection{Design of Truncated Optimal Additive Noise}\label{subsection:Design-of-Truncated-Optimal-Additive}

Additive noise mechanisms such as the Gaussian and Laplace mechanisms typically perturb the measurements with particular random noise to achieve DP. 
However, most of the results require the support of the noise distribution to be unbounded, except for a few cases such as \cite{Liu_2019,Croft2022}. This class of mechanisms is not 
% Due to the nature of set-based estimation and safety considerations, they \tr{unclear reference} are not directly 
applicable for some applications, such as safety-critical systems \cite{DBLP:journals/corr/abs-2107-12957}, which require state enclosure guarantees in a bounded set. One closely related study \cite{9147726} developed an interval observer with DP based on a truncated Laplace mechanism. However, the design relies on an $L_1$ norm-based adjacency relation and an analytical bound for the noise variance, which may cause the design to be conservative, consuming more noise to achieve DP (a numerical comparison is provided in Section \ref{sec:eval}).
To solve this issue, we follow the numerical approach, recently developed in \cite{DBLP:journals/corr/abs-2107-12957}, to optimize a noise distribution that is subjected to a bounded support constraint and the privacy constraint in Definition \ref{def:approximate-dp}.

In our setups, all sets are defined over the continuous domain of all real numbers. However, the truncated optimal noise distribution we employ to achieve DP is generated numerically. This distribution consists of discrete noise occurrence probabilities, with noise values selected from the continuous domain. Consequently, even after adding this noise to the sets, the resultant domain remains continuous.
% However, the truncated optimal noise distribution that we employ to achieve DP is a distribution that is generated numerically and consists of discrete noise occurrence probabilities with noise values selected from the continuous domain. Hence, the resultant domain of adding that noise to the sets is still continuous. 

Next, we define a class of truncated noise distribution functions. Then, for a fixed DP parameter $\epsilon$ and a particular noise model, we present an optimization problem, where the objective function 
% is the loss function in \eqref{eqn:optimization-function}, which 
balances between the privacy loss parameter $\delta$ and the utility loss. 
Upon solving this optimization problem, an optimal noise distribution is generated. 

\begin{definition}[Truncated Noise Distribution \cite{DBLP:journals/corr/abs-2107-12957}]\label{def:optimal-noise}
Let ${\Phi}=[-d,d]$ define a bounded noise range such that $d \in \mathbb{R}$, and $\Acute{\Phi}=\{\phi_l\}_{l\in \{1,\dots,\;2N\}}$ be the discretization of $\Phi$ on $2N$ equidistant steps such that $\phi_l\in \Phi$, $N\in \mathbb{Z}^+$, and $l\in \{1,\dots,\;2N\}$,
then a numerically generated, truncated, discrete, equidistant, symmetric, and monotonically decreasing from its zero center noise distribution function, denoted by $P(\phi_l)$, has the following properties
% is defined as 
\begin{subequations}\label{eqn:optimal-noise-function}
\begin{equation}
    \sum_{\phi_l\in \Acute{\Phi}} P(\phi_l)=1\;\; \mathrm{and}\;\; P(\phi_l)\geq 0. \;\;\;\; \;\;(\textrm{Distribution})
\end{equation}    
    \begin{equation}
    P(\phi_l) \geq P(\phi_m)\;\;\; \forall\; \phi_m>\phi_l>0, \;\;\;\;(\textrm{Monotonicity}) 
    \end{equation}
\quad \qquad where $l,m\in \{N+1,\dots,\;2N\}$.
    \begin{equation}
   \qquad\quad P(\phi_l) = P(-\phi_l),
    \quad\qquad\qquad
    % \;\;\;\;\;\;\; \;\;\;\;\;\;\;\;\;\; 
    (\textrm{Symmetry})      
    \end{equation}
    \quad \qquad where $l\in \{1,\dots,\;N\}$.
    \end{subequations}
% while satisfying the $(\epsilon, \delta)$-ADP constraints.
\end{definition}

For the LDP model, a dataset of interest is the local measurement signal $y^{(i)}\in \mathcal{H}$ released by sensor $i$ as a list of measurements $\{y^{(i)}_k\}_{k\in\{1, \dots, \mathcal{T} \}}$, where $\mathcal{T}=\infty$ is also of interest \cite{9147726}. We aim to provide a certain level of privacy protection for a single sensor's measurement. Two local measurement signals $y^{(i)}$ and $\Acute{y}^{(i)}$ are called adjacent and can be denoted by Adj$(y^{(i)},\Acute{y}^{(i)})$, if and only if they differ by the value of exactly one measurement $y_k^{(i)}$ \cite{degue2017differentially}. In other words, $y^{(i)}$ and $\Acute{y}^{(i)}$ are considered adjacent if there is any time step $k$ at which $y_k^{(i)}\neq \Acute{y}_k^{(i)}$. Since the privacy-preserving noise is added to the measurements themselves, then the allowed variation within Adj$(y^{(i)},\Acute{y}^{(i)})$ is bounded by what is called sensitivity, which is defined formally as follows.

\begin{definition} [Sensitivity - LDP \cite{9147726,6483414}]
\label{def:sensitivity-ldp}
For a given sensor $i$, the allowed deviation for a single measurement $y_k^{(i)}$ between two adjacent local measurement signals $y^{(i)}$ and $\Acute{y}^{(i)}$, i.e., Adj$(y^{(i)},\Acute{y}^{(i)})$, is bounded in the $L_2$ norm by $s$ and given by 
\begin{align}
\lVert y^{(i)}-\Acute{y}^{(i)}\rVert_2 \leq s,
\end{align}
where $s\geq0$.
\end{definition}

Next, the following definition describes the additive noise mechanism for the LDP setup, shown in Figure \ref{fig:local-DP-Setting}, in which each participating sensor locally perturbs its own measurement, and the sensor manager is not present in this setup since it is considered an untrusted entity.  
% \begin{figure}[bhpt]
% \graphicspath{ {./Figures/} }
%     \centerline{
% \includegraphics
% [scale=0.4]{Figures/cloud-estimator-ldp}}
%     \caption{Cloud-based estimator using the LDP setup.}
%     \label{fig:cloud-estimator-ldp}
% \end{figure}

\begin{definition}[Additive Noise Mechanism - LDP \cite{DBLP:journals/corr/abs-2107-12957}
]\label{def:additive-noise-mechanism-ldp-measurement}
Given a measurement of sensor $i$ $y^{(i)}_k$ and a noise sample $\phi_k \in \Acute{\Phi}$ such that the successive samples of $\phi_k$ are IID with the probability distribution in Definition \ref{def:optimal-noise} and satisfying Lemma \ref{lemma:dp-optimal-noise-mechanism}, then the additive noise mechanism $M^{(i)}_y$ is defined as
\begin{equation}\label{eqn:additive-noise-mechanism-ldp-measurement}
    {M}^{(i)}_y:\hat{y}^{(i)^{DP}}_k = y^{(i)}_k + \phi_k, 
\end{equation}
where %$\hat{y}_k$ is a vector of the anonymized measurements.
$\hat{y}^{(i)^{DP}}_k$ is a protected measurement.
\end{definition}

\begin{lemma}[{\cite [Theorem 15]{DBLP:journals/corr/abs-2107-12957}}]\label{lemma:dp-optimal-noise-mechanism}
Let $M^{(i)}_{y}$ be an additive noise mechanism with a sensitivity $s$ (Definition \ref{def:sensitivity-ldp})
% , operating on a vector of measurements $y_k$,
and $\Acute{\Phi}=\{\phi_l\}_{l\in \{1,\dots,\;2N\}}$ be the discretization of ${\Phi}$ with the truncated optimal noise distribution $P(\phi_l)$ (Definition \ref{def:optimal-noise}).
% Let $d \in \{\phi_l - \phi_m|l,m \in \mathbb{Z}\}, d\leq s$
If $\forall\;\hat{\Phi}\subseteq \Acute{\Phi}$
\begin{equation}\label{eqn:dp-optimal-noise}
    \sum_{\phi_l\in \hat{\Phi}}P(\phi_l) \leq e^{\epsilon}{\sum_{\phi_l\in \hat{\Phi}}P(\phi_l+s)}+\delta, \end{equation}
then
%for any vector of measurements $y_k\in \mathbb{D}$, 
the additive noise mechanism $M^{(i)}_{y}$ is $(\epsilon, \delta)$-ADP for any $y^{(i)}_k\in \mathcal{H}$.
%, and the vector of the anonymized measurements, denoted by $\hat{y}^{DP}_k$, also becomes a vector of $(\epsilon, \delta)$-ADP measurements.
\end{lemma}

It deserves noting that, given that the sensitivity $s$ (Definition \ref{def:sensitivity-ldp}) is satisfied $\forall\;k\in \mathbb{Z^+}$, the mechanism $M^{(i)}_y$ is $(\epsilon, \delta)$-ADP at any time step $k$ \cite[Lemma 2]{6606817}. 

    Next, we present the loss function, 
    denoted by $L^{\Omega_t}_{\gamma}$ \cite{DBLP:journals/corr/abs-2107-12957}. % \st{, for learning the optimal noise distribution $P(\phi_l)$}
    This function balances between the privacy parameter $\delta$ and a utility loss $U$ at a fixed $\epsilon$ and is given by \begin{subequations} \label{eqn:optimization-function}
\begin{equation} \label{eqn:optimization-function-1}    L^{\Omega_t}_{\gamma} = \delta+\Omega_t U,
\end{equation}
where the utility loss $U$ is given by
\begin{equation}
  U=\Bigg(\sum_{\phi_l \in \Acute{\Phi}} |\phi_l|^\gamma P(\phi_l)\Bigg)^{1/\gamma}   
\end{equation}
with $\gamma\in\{1,2\}$ such that $\gamma$ selects between $L_1$ or $L_2$ norm-based utility loss, and
\begin{equation}
    \Omega_t= \max\bigg(\frac{\Omega_{start}}{2^{t/\Gamma}},\Omega_{min}\bigg)
\end{equation}
\end{subequations}
is the utility weight at training epoch $t$, where $t\in \mathbb{Z^+}$ with an exponentially decaying rate $\Gamma$ from a starting value $\Omega_{start}$ and with a lower bound $\Omega_{min}$. The optimal noise distribution $P(\phi_l)$ is then optimized by minimizing the weighted sum of the utility loss $U$ and the privacy parameter $\delta$. This noise distribution is generated through the following steps, and we are going to omit $(\phi_l)$ to ease the notation. 
We start by generating the first monotonically increasing half (i.e., $\{\phi_l\}_{l\in \{1,\dots,\; N\}}$) of the noise distribution $P(\phi_l)$, which is given by
\begin{subequations} \label{eqn:noise-model}
\begin{equation}\label{eqn:noise-model-2}
    P_l=1/2 \;\textrm{SoftMax}(r_l);\;r_l \in \{r_0,\dots,\;r_N\},
\end{equation}
where SoftMax($r_l$) = $e^{r_l}/\sum^{N}_{i=0}e^{r_i}$ and $N$ is the number of discretization steps in the half-width of the noise distribution $P(\phi_l)$. The SoftMax function normalizes the $r_l$ values into a distribution, and the $r_l$ values are generated using a model of $v$-stacked Sigmoid functions (i.e., $\sigma(\phi_l)=(1+e^{-\phi_l})^{-1}$), which is given by
\begin{equation}\label{eqn:noise-model-1}
    r_l= \ln \bigg[A^2+\sum_{j=0}^v B_{j}^2\cdot \sigma(C\cdot(\phi_l-F_j)) \bigg],
\end{equation}
where $\phi_l=l\frac{d}{N}-d$. The parameters $A,\; B_j,\;C$, and $F_j$ are randomly initialized, then learned to optimize the loss function in \eqref{eqn:optimization-function} using numerical optimization methods (e.g., stochastic gradient descent (SGD)).

Next, the first half (i.e., $\{\phi_l\}_{l\in \{1,\dots,\; N\}}$) of the noise distribution generated by \eqref{eqn:noise-model-2} and denoted by $P_l$ is mirrored according to the following:
\begin{equation}
    P_j=P_{2N-j+1}\;\; \mathrm{for}\; j\in\{N+1,\cdots,\;2N\}. 
\end{equation}
Then, $P_j$ (i.e., $\{\phi_l\}_{l\in \{N+1,\dots,\;2N\}}$) is concatenated to $P_l$ (i.e., $\{\phi_l\}_{l\in \{1,\dots,\;N\}}$) to obtain the symmetric noise distribution $P(\phi_l)$ (i.e., $\{\phi_l\}_{l\in \{1,\dots,\;2N\}}$). 
% as shown in Figure \ref{fig:optimal-noise-distribution}.
\end{subequations}

Note that the case with arbitrarily dimensional and spherically rotation-symmetric noise distributions and sensitivity conditions can be reduced to a 1-dimensional privacy analysis \cite{10.1145/2976749.2978318}. Based on this claim, for the CDP setup, we can define the additive noise for the case in which the dataset of interest is the global measurement signal $y \in \mathcal{H}$ released by an array of $m$ sensors as a list of measurements vectors $\{y_k\}_{k\in\{1, \dots, \mathcal{T} \}}$, where $\mathcal{T}=\infty$ is also of interest and $y_k=\big[{y_k^{(1)}},\dots,{y_k^{(m)}}\big]^T$ \cite{9147726}. Two global measurement signals $y$ and $\Acute{y}$ are called adjacent and can be denoted by Adj$(y,\Acute{y})$, if and only if they differ by the value of exactly one measurement vector $y_k$. Similarly, the allowed variation within Adj$(y,\Acute{y})$ is bounded by what is called sensitivity, which is defined formally as follows.

\begin{definition} [Sensitivity - CDP \cite{9147726,6483414}]
\label{def:sensitivity-cdp}
The allowed deviation for a single measurement vector between two adjacent global measurement signals $y$ and $\Acute{y}$, i.e., Adj$(y,\Acute{y})$, is bounded in the $L_2$ norm by $s_{g}$ and given by 
\begin{align}
\lVert y-\Acute{y}\rVert_2 \leq s_{g},
\end{align}
where $s_{g}\geq0$.
\end{definition}  

Next, for the CDP setup, shown in Figure \ref{fig:Centralized-DP-Setting}, the sensor manager aggregates the measurements of all sensors into a measurement vector and perturbs them with the privacy-preserving noise according to the following additive noise mechanism.
% \begin{figure}[bhpt]
% \graphicspath{ {./Figures/} }
%     \centerline{
% \includegraphics
% [scale=0.4]{Figures/cloud-estimator-cdp}}
%     \caption{Cloud-based estimator using the CDP setup. }
%     \label{fig:cloud-estimator-cdp}
% \end{figure}

\begin{definition}[Additive Noise Mechanism - CDP
 \cite{DBLP:journals/corr/abs-2107-12957}] \label{def:additive-noise-mechanism}
Given a vector of measurements $y_k$ and a noise vector $\Phi_k \in \Acute{\Phi}$ of independent and identically distributed (IID) coordinates with the probability distribution in Definition \ref{def:optimal-noise} and satisfying Lemma \ref{lemma:dp-optimal-noise-mechanism} using the sensitivity $s_{g}$ (Definition \ref{def:sensitivity-cdp}), and assuming that successive samples of $\Phi_k$ are also IID, then the additive noise mechanism $M_y$ is defined as
\begin{equation}\label{eqn:additive-noise-mechanism}
    {M}_y:\hat{y}^{DP}_k = y_k + \Phi_k, 
\end{equation}
where %$\hat{y}_k$ is a vector of the anonymized measurements.
$\hat{y}^{DP}_k$ is a vector of protected measurements.
\end{definition}
Likewise, given that the sensitivity $s_{g}$ (Definition \ref{def:sensitivity-cdp}) is satisfied $\forall\;k\in \mathbb{Z^+}$, the mechanism $M_y$ is $(\epsilon, \delta)$-ADP at any time step $k$ \cite[Lemma 2]{6606817}. 

%It is worth noting that the protected measurements obtained within the context of the LDP and CDP models serve as the centers of the measurement sets. Therefore, we have measurement sets with anonymized centers. Next, we will deploy these measurement sets to the set-based estimator.
\subsection{Differentially Private Zonotope-based Set-Membership Estimation}\label{subsection:Differentially-Private-Zonotope-based-Set-Membership-Estimation}
 In this subsection, we introduce the differentially private set-based estimator. The corrupted measurements serve as the centers of the measurement sets. Hence, we have measurement sets with corrupted centers that guarantee differential privacy for the contained sensitive measurements. The generators of all zonotopes are not corrupted with the privacy-preserving noise. 
 %After obtaining the protected measurements within the context of the CDP or LDP model,
 We will deploy these measurement sets to the set-based estimator (Algorithm \ref{alg:set-estimation-using-zonotope}), i.e., the cloud estimator, to preserve the privacy of the sensitive measurements during the estimation process and get state estimates while keeping the measurements safeguarded against any untrusted entity. Hence, a differentially private version of that set-based estimator is summarized in the following algorithm.

\begin{algorithm}
\caption{($\epsilon, \delta$)-ADP Set-Estimation Using Zonotopes}\label{alg:adp-set-estimation-using-zonotope}
\begin{algorithmic} [1]
{\Statex {\bfseries Input:} $\epsilon$, $d$, sensitivity, $\mathcal{Z}_p= \zono{0, G_p}$, $\mathcal{Z}_w= \zono{0, G_w}$, and $\mathcal{Z}^{(i)}_{v}= \zono{0, G^{(i)}_{v}}$, $i \in \{1, \dots, m\}$.
} 
{ \Statex {\bfseries Output:} $\mathcal{\Bar{Z}}^{DP}_{k}=\zono{\Bar{c}^{DP}_k, \Bar{G}^{DP}_k}$.}
\Statex{\bfseries Initialization:} The truncated optimal noise distribution $P(\phi_l)$ is generated. Set $k = 1$ and $\mathcal{\Bar{Z}}_{0} = \zono{\Bar{c_0},\Bar{G_0}}$.
\While{$True$}
\State 
% $m$ sensitive measurements $y_k^{(i)}$, $ i\in \{1,\dots,\; m\}$ provided by an array of $m$ sensors.
Obtain $y_k^{(i)}$ from each sensor $i$, $ i\in \{1,\dots,\; m\}$. \label{line:2-start-add-noise}
% An array of $m$ sensors, where each sensor provides a sensitive measurement $y_k^{(i)}$, $ i\in \{1,\dots,\; m\}$.
\If{the setup in Figure \ref{fig:Centralized-DP-Setting} is deployed}
\State The sensor manager aggregates the $m$ measurements into a vector of measurements $y_k=\big[{y_k^{(1)}},\dots,{y_k^{(m)}}\big]^T$, then adds the DP noise to $y_k$ to obtain protected measurements $\hat{y}^{DP}_k$ following \eqref{eqn:additive-noise-mechanism}.
% \algstore{myalg-dp}
% \end{algorithmic}
% \end{algorithm}
% \begin{algorithm}
% \begin{algorithmic} [1]
% \algrestore{myalg-dp}
\ElsIf{the setup in Figure \ref{fig:local-DP-Setting} is deployed}
\State Each sensor locally perturbs its measurement to obtain a protected measurement $\hat{y}^{(i)^{DP}}_k$ according to \eqref{eqn:additive-noise-mechanism-ldp-measurement}.
\EndIf \label{line:8-end-add-noise}

\State The cloud estimator uses 
$\mathcal{Z}^{(i)^{DP}}_{y_k}=\zono{y^{(i)^{DP}}_k, G^{(i)}_{v}}, i\in \{1, \dots, \; m\}$ to perform the estimation process as follows:
\State Compute the predicted state set $ \mathcal{\hat{Z}}^{DP}_{k}= \zono{\hat{c}^{DP}_k, \hat{G}^{DP}_k}$ as follows:
\begin{align}
    \hat{c}^{DP}_k =&f({x_{k-1}^{*}})+\frac{\partial f_{k-1}}{\partial x}|_{x_{k-1}^{*}}\Big(\hat{c}^{DP}_{k-1} {-} x_{k-1}^{*}\Big) {+} c_{L,k},\label{eqn:dp-time-update-zonotope-center-non-linear}
\end{align}
\begin{align}
    \hat{G}^{DP}_k =&\bigg[\frac{\partial f_{k-1}}{\partial x}|_{x_{k-1}^{*}}\hat{G}^{DP}_{k-1}, G_{L,k}, G_w\bigg] .\label{eqn:dp-time-update-zonotope-generator-non-linear}
\end{align}
\label{line:10-predicted-state-set} 
\State Compute the corrected state set $\mathcal{\Bar{Z}}^{DP}_{k}=\zono{\Bar{c}^{DP}_k, \Bar{G}^{DP}_k}$ in lines \ref{line:12} to \ref{line:17}.
%in \eqref{eqn:dp-estimation-with-zonotope-center-non-linear}, \eqref{eqn:dp-estimation-with-zonotope-generator-non-linear}.
\State\begin{align}
    \Acute{c}^{DP}_k =&  \hat{c}^{DP}_{k} + \sum_{i=1}^{m} \lambda^{(i)}_{k} \bigg( y^{(i)^{DP}}_{k} - h^{(i)}\big({x_{k}^{*}}\big)\nonumber\\ 
&-\frac{\partial h_{k}^{(i)}}{\partial x}|_{x_{k}^{*}}(\hat{c}^{DP}_{k} {-} x^{*}_{k}) {-} c_{L,k} 
% -c_{p} - c^{(i)}_{v} 
\bigg),\label{eqn:dp-estimation-with-zonotope-center-non-linear}
\end{align}
%$\Acute{c}^{DP}_k =  \hat{c}^{DP}_{k} + \sum_{i=1}^{m} \lambda^{(i)}_{k} \big ( y^{(i)^{DP}}_{k} - h^{(i)}({x_{k}^{*(i)}}) $\par$\quad\quad\;-\frac{\partial h_{k}^{(i)}}{\partial x}|_{x_{k}^{*(i)}}(\hat{c}_{k} {-} x^{*(i)}_{k}) {-} c^{(i)}_{L,k} \big)$, 

% \If{The model is Linear \eqref{eqn:system-model} }
% \State $\Acute{c}^{DP}_k =  \hat{c}^{DP}_{k} + \sum_{i=1}^{m} \lambda^{(i)}_{k} \big ( \hat{y}^{(i)^{DP}}_{k} - H^{(i)} \hat{c}^{DP}_{k} - c^{(i)}_{v}-c_{p} \big)$, 
% \State $\Acute{G}^{DP}_{k} =\bigg[(I- \sum_{i=1}^{m+1} \lambda^{(i)}_{k} H^{(i)})\hat{G}^{DP}_{k},$\par
% $\quad \quad \quad \; \; \lambda^{(1)}_{k} G^{(1)}_{v},\dots, \lambda^{(m)}_{k} G^{(m)}_{v}, \lambda^{(m+1)}_{k}{G}_{p}\bigg]$.
% \ElsIf{The model is non-linear \eqref{eqn:system-model-non-linear}}
\label{line:12}
\State  
\begin{align}
    \Acute{G}^{DP}_k =&\bigg[\bigg(I- \sum_{i=1}^{m} \lambda^{(i)}_{k} 
    \frac{\partial h_{k}^{(i)}}{\partial x}|_{x_{k}^{*}}
    \bigg) \hat{G}^{DP}_{k},\;-\lambda^{(1)}_{k} G_{L,k},\dots, \;\nonumber\\& -\lambda^{(m)}_{k} G_{L,k},-\lambda^{(1)}_{k}G_{p},\dots,-\lambda^{(m)}_{k}G_{p},\;\nonumber\\& -\lambda^{(1)}_{k}G^{(1)}_{v}, \dots, -\lambda^{(m)}_{k}G^{(m)}_{v}  \bigg].\label{eqn:dp-estimation-with-zonotope-generator-non-linear}
\end{align}
% where $G_{p}$ is the generator matrix, created using the range of the optimal noise $d$.
%$\Acute{G}^{DP}_k =\Big[\big(I- \sum_{i=1}^{m} \lambda^{(i)}_{k} \frac{\partial h_{k}^{(i)}}{\partial x}|_{x_{k}^{*(i)}} \big) \hat{G}^{DP}_{k},\;-\lambda^{(i)}_{k} G^{(i)}_{L,k}, $\par$\quad \quad \; \; -\lambda^{(i)}_{k}G_{p},\; -\lambda^{(1)}_{k}G^{(1)}_{v}, \dots, -\lambda^{(m)}_{k}G^{(m)}_{v}  \Big]$.
% \EndIf 
\State Compute the weights $\Bar{\Lambda}_k^*$ as follows:
\begin{equation}\label{eqn:dp-optimal-lambda-weights}
\Bar{\Lambda}_k^* = \argminB_{{\lambda}_k}\lVert \Acute{G}^{DP}_k \rVert_{F}^{2},
\end{equation}
where $\Bar{\Lambda}_k^* = [ {\lambda}^{(1)}_{k},\dots,\; {\lambda}^{(m)}_{k}]$. \label{line:14-optimal-weights} 

% in \eqref{eqn:dp-optimal-lambda-weights}. 
% $= \argminB_{{\lambda}_k}\lVert \Acute{G}^{DP}_k \rVert_{F}^{2}$.
 
\State Reduce the order of $\Acute{G}^{DP}_k$ 
as in \cite{10.1007/978-3-540-31954-2_19}: $\Bar{G}^{DP}_{k} =\downarrow_q	\Acute{G}^{DP}_k$. 
\State $\Bar{c}^{DP}_k = \Acute{c}^{DP}_k$.
\State $\mathcal{\Bar{Z}}^{DP}_{k}= \zono{\Bar{c}^{DP}_k, \Bar{G}^{DP}_k}$. \label{line:17}
\State Update the time: $k=k+1$. \label{line18:time-step}
\EndWhile
\end{algorithmic}
\end{algorithm}
Algorithm \ref{alg:adp-set-estimation-using-zonotope} summarizes our proposed differentially private set-based estimator. The inputs to the algorithm are the DP parameter $\epsilon$, the noise range $d$, the sensitivity, the privacy-preserving noise zonotope $\mathcal{Z}_p= \zono{0, G_p}$, where $G_{p}$ is the generator matrix, created using the range of the optimal noise $d$, the process noise zonotope $\mathcal{Z}_w= \zono{0, G_w}$, and the measurement noise zonotopes $\mathcal{Z}^{(i)}_{v}= \zono{0, G^{(i)}_{v}}$, where $i \in \{1, \dots, m\}$. At the initialization, the truncated optimal noise distribution $P(\phi_l)$ is generated using $\epsilon$, $d$, and sensitivity according to \eqref{eqn:optimal-noise-function} while optimizing the loss function in \eqref{eqn:optimization-function} by learning the parameters of the noise model in \eqref{eqn:noise-model}. Then, the privacy-preserving noise is added to the sensitive measurements using either of the two additive noise mechanisms described in Definitions \ref{def:additive-noise-mechanism-ldp-measurement} and \ref{def:additive-noise-mechanism} in lines \ref{line:2-start-add-noise} to \ref{line:8-end-add-noise}. The cloud estimator uses the $m$ measurement sets containing the protected measurements $\mathcal{Z}^{(i)^{DP}}_{y_k}=\zono{y^{(i)^{DP}}_k, G^{(i)}_{v}}$ to perform the estimation process. The cloud estimator computes the predicted state set $ \mathcal{\hat{Z}}^{DP}_{k}$ in line \ref{line:10-predicted-state-set}, where the center of the state, $x_{k-1}^{*}$, is used as a linearization point for the state function $f$ at time step $k-1\in \mathbb{Z}^+$. The infinite Taylor series is over-approximated by the first-order Taylor series and its Lagrange remainder $\mathcal{{Z}}_{L,k} = \zono{c_{L,k}, G_{L,k}}$. The corrected state set  $ \mathcal{\Bar{Z}}^{DP}_{k}$ is computed in lines \ref{line:12} to \ref{line:17}, where the center of the state, $x_{k}^{*}$, is used as a linearization point at time step $k\in \mathbb{Z}^+$. The infinite Taylor series is over-approximated by the first-order Taylor series and its Lagrange remainder $\mathcal{{Z}}_{L,k} = \zono{c_{L,k}, G_{L,k}}$ to guarantee true state inclusion. In line \ref{line:14-optimal-weights}, as in \cite{conf:disdiff}, the weights $\Bar{\Lambda}_k^*$ are optimized to reduce the Frobenius norm of the generator matrix $ \Acute{G}^{DP}_k$; therefore, they reduce uncertainty around estimated values. The time is updated for the next time step in line \ref{line18:time-step}.
\begin{theorem}\label{theorem:dp-estimator-non-linear} Given that the two additive noise mechanisms described in Definitions \ref{def:additive-noise-mechanism-ldp-measurement} and \ref{def:additive-noise-mechanism} are $(\epsilon, \delta)$-ADP. Then, Algorithm \ref{alg:adp-set-estimation-using-zonotope} guarantees the true state containment in the estimated set and differential privacy for the sensitive measurements throughout the set-based state estimation process.

% is a differentially private set-based estimator that obtains state estimates while preserving the privacy of the sensitive measurements.
% Given $m$ measurement sets containing the protected measurements. Then the differentially private set-based estimator obtains state estimates in a privacy-preserving manner as follows

% Given $m$ measurement sets containing the protected measurements, a differentially private version of the set-based estimator (Lemma \ref{lemma-set-based-estimation-zonotope-non-linear}) obtains state estimates in a privacy-preserving manner as follows
%, which obtains state estimates while preserving the privacy of the measurement sets is defined as 

% Given an $(\epsilon,\delta)$-ADP additive noise mechanism %$M_y$
% (Definition \ref{def:additive-noise-mechanism}, \ref{def:additive-noise-mechanism-ldp-measurement}) with sensitivity $s$ (Definition \ref{def:sensitivity}) and the truncated optimal noise distribution $P(\phi_l)$ satisfying \eqref{eqn:optimal-noise-function}, \eqref{eqn:optimization-function}, \eqref{eqn:noise-model}, and Lemma \ref{lemma:dp-optimal-noise-mechanism}, then using a vector of $(\epsilon,\delta)$-ADP measurements $\hat{y}^{DP}_k$, an approximate differentially private $(\epsilon,\delta)$-ADP version of the set-based estimator (Lemma \ref{lemma-set-based-estimation-zonotope-non-linear}) which obtains state estimates that are insensitive to a single sensor’s measurement can be defined as 

\end{theorem}

\begin{proof}

The state estimates are obtained while the sensitive measurements are kept DP-protected since the differential privacy guarantees are preserved by post-processing \cite[Proposition~2.1]{TCS-042}. The true state containment in the estimated set follows from the boundness of the process noise, the measurement noise, and the privacy-preserving noise, which are all added as a Minkowski sum to the estimated set.
% We add a zonotope of the bounded privacy-preserving noise as a Minkowski sum to the estimated set. The generator of this zonotope is created using the range of optimal noise, representing the worst-case value. This guarantees the containment of the true state.
\end{proof}

%% file: Sections/5-eval.tex
\section{Example} \label{sec:eval}
% In this section, we evaluate the proposed $(\epsilon,\delta)$-ADP differentially private set-based estimator for the CDP and LDP setups. We present the experimental results for the Matlab R2023b implementation of our proposed $(\epsilon,\delta)$-ADP set-based estimator. For zonotope operations, we use the CORA toolbox.
This section evaluates the proposed differentially private set-based estimator in both the LDP and CDP setups visualized in Figure \ref{fig:local-DP-Setting} and Figure \ref{fig:Centralized-DP-Setting}, respectively. Experimental results are presented based on the implementation in MATLAB R2023b, with zonotope operations performed using the CORA toolbox \cite{conf:cora1}. We evaluate the differentially private set-based estimator through the localization of a quadcopter navigating through arbitrary non-linear motion within three-dimensional space measuring $10\; \times\; 10\; \times\; 10\;m^3$. 
According to the following model, the anchor nodes provide measurements as relative distances to the intruding quadcopter. The proposed differentially private set-based estimator can handle the bounded measurement uncertainties and anonymize these measurements with privacy-preserving noise, thus concealing the exact locations of the anchor nodes from the cloud estimator or any untrusted party. This localization is achieved using a high rate of real-world data measurements provided by a set of $8$ anchor nodes distributed across the motion area \cite{conf:d-slats}. The model of this system incorporates a linear state function represented by the matrix $f$, which is given by
\begin{equation*}
%  F=\begin{bmatrix}
% 1 & 0 & 0\\
% 0 & 1 & 0\\
% 0 & 0  & 1
% \end{bmatrix},\;
f=\text{diag}([1\; 1 \; 1]^T).
\end{equation*}
% The measurement matrix in this example is replaced by the nonlinear
Additionally, the measurement function $h^{(i)}(x_k)$ is defined as
\begin{equation*}
    h^{(i)}(x_k)=\lVert x_A^{(i)}-x_k\rVert_2, 
\end{equation*}
% function $h^{(i)}(x_k)=\lVert x_A^{(i)}-x_k\rVert\;_2 $, 
where $x_A^{(i)}$ represents the location of anchor node $i$ and $x_k$ is an estimate for the location of the quadcopter at time step $k$. 
The measurement and process noise zonotopes $\mathcal{Z}^{(i)}_{v}$ and $\mathcal{Z}_{w}$, respectively, are set to: 
\begin{equation*}
    z^{(i)}_{v}=\Big\langle{ 0 ,\begin{bmatrix} 0.01\;\; 0.02\;\; 0.01 \end{bmatrix}} \Big\rangle
,\end{equation*}
\begin{equation*}
    z_{w}=\Big\langle{[0 \; 0\; 0]^T ,\text{diag}([0.50 \;0.50 \; 0.50 ]^T)
    % \begin{bmatrix} 0.50 & 0 & 0 \\ 0 & 0.50 & 0\\ 0 & 0 & 0.50 \end{bmatrix}
    } \Big\rangle.
\end{equation*}

\subsection{The CDP Model}\label{Evaluation-for-The-CDP-Setup}
% In this subsection, we evaluate the $(\epsilon,\delta)$-ADP differentially private set-based estimator within the context of the CDP setup using two localization examples.
% The first example is %Consider
% a network of $8$ sensory nodes tracking the location of an object in a circular motion within two-dimensional space with dimensions of $180\; \times 180\; m^2$ \cite{conf:disdiff}.
In this subsection, we evaluate the differentially private set-based estimator within the context of the CDP setup.
% The first example involves a network of 8 sensory nodes tracking the location of an object moving in a circular motion within a two-dimensional space measuring $180\; \times 180\; m^2$ \cite{conf:disdiff}.
% We aim to keep the sensory nodes' measurements private from the cloud estimator while estimating the location of the rotating object.
% % The state of each node consists of two variables representing the position of the rotating object. 
% The state of each node comprises two variables representing the position of the rotating object. 
% The model of this system encompasses linear state and measurement functions, represented by the matrices $F$ and $H^{(i)}$, respectively, as follows:
% % The state and measurement matrices are, respectively, as follows: 
% \begin{equation*}
%  F=\begin{bmatrix}
% 0.9920 & -0.1247 \\
% 0.1247 & 0.9920 
% \end{bmatrix},\;
% H^{(i)}=\begin{bmatrix}
% 1 &0 \\ 
% \end{bmatrix}\; \mathrm{or}\;
% \begin{bmatrix}
% 0 &1 
% \end{bmatrix},
% \end{equation*}
% where the measurement matrix $H^{(i)}$ is selected based on the measurement sequence. The measurement and process noise zonotopes $\mathcal{Z}^{(i)}_{v}$ and $\mathcal{Z}_{w}$ are, respectively, set to: 
% \begin{equation*}
%     z^{(i)}_{v}=\Big\langle{\begin{bmatrix} 0 \end{bmatrix},\begin{bmatrix} 0.01\;\; 0.02 \end{bmatrix}} \Big\rangle
% ,\;
%     z_{w}=\Bigg\langle{\begin{bmatrix} 0 \\0 \end{bmatrix},\begin{bmatrix} 0.50 & 0 \\0 & 0.50 \end{bmatrix}} \Bigg\rangle.
% \end{equation*}
The $(\epsilon,\delta)$-ADP truncated optimal noise distribution (Definition \ref{def:optimal-noise}) is generated according to the noise model in \eqref{eqn:noise-model} with $C=500$. 
% in order to allow sudden jumps for the noise function $P(\phi_l)$.
We set $d\in[-1,\;1]$, $\epsilon=0.3$, and $s=1$ and the parameters $A$, $B_j$, $F_j$, and $\delta$ are learned to optimize the loss function in \eqref{eqn:optimization-function} using the SGD tool \cite{DBLP:journals/corr/abs-2107-12957}.

\begin{figure*}[h]
%\vspace{-0.05cm}
    %\centering
    \begin{tabular}{ p{0.30\textwidth}  p{0.3\textwidth}  p{0.3\textwidth}}
        \resizebox{0.31\textwidth}{!}{
            \begin{subfigure}[h]{0.4\textwidth}
      \centering
        \includegraphics[scale=0.35]{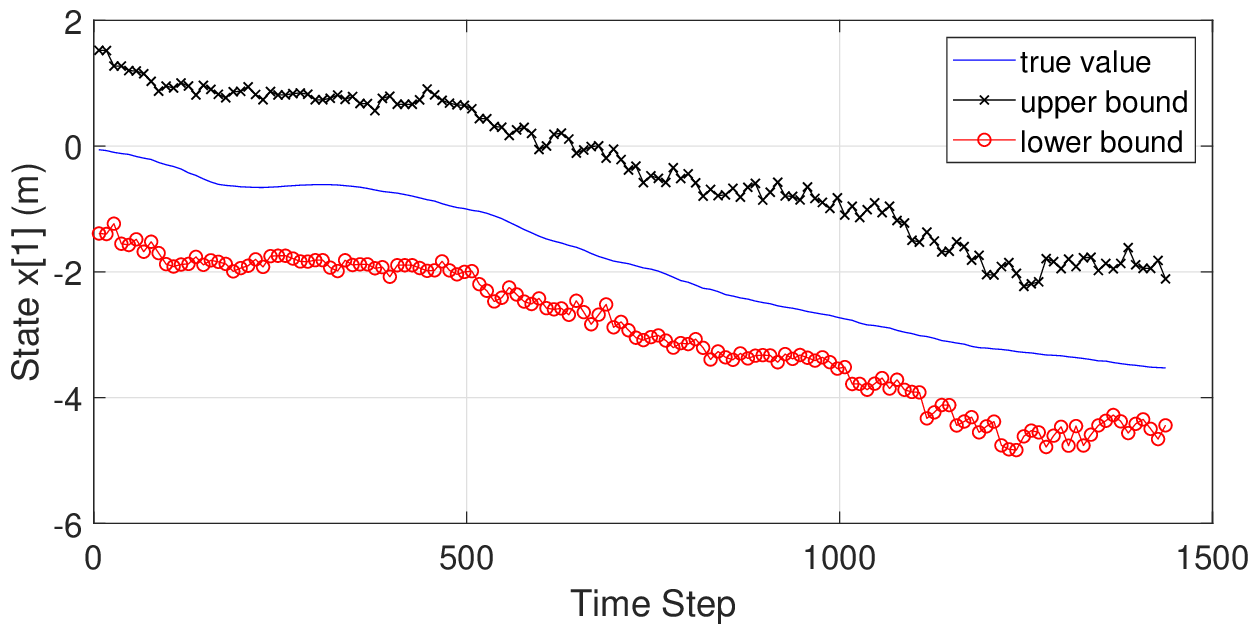}
        \caption{}
        \label{fig:cdp-statex1}
    \end{subfigure}
       } 
   &
   \resizebox{0.31\textwidth}{!}{
            \begin{subfigure}[h]{0.4\textwidth}
      \centering
        \includegraphics[scale=0.35]{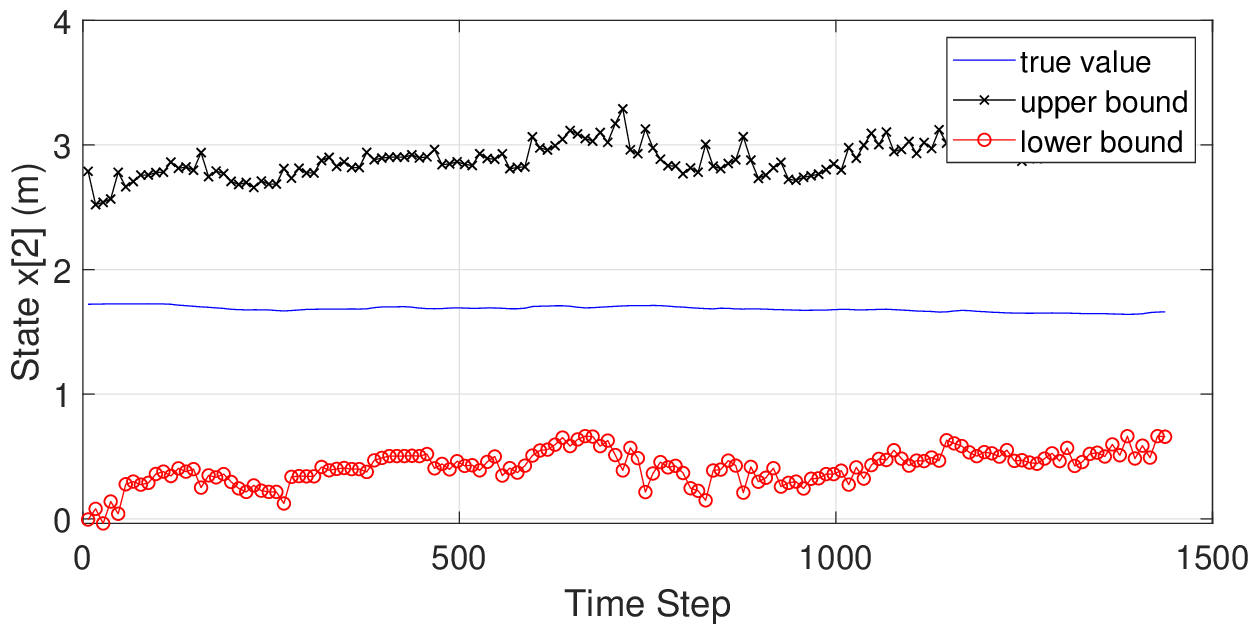}
        \caption{}
        \label{fig:cdp-statex2}
    \end{subfigure}
      }
 &
   \resizebox{0.31\textwidth}{!}{
            \begin{subfigure}[h]{0.4\textwidth}
      \centering
        \includegraphics[scale=0.35]{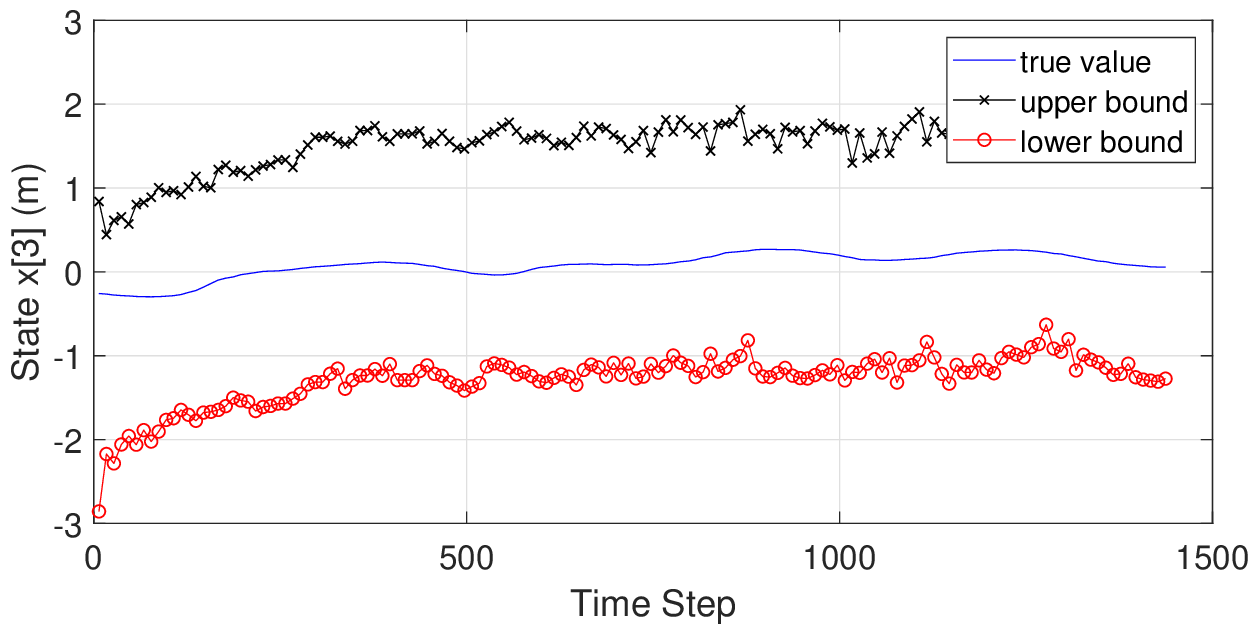}
        \caption{}
        \label{fig:cdp-statex3}
    \end{subfigure}
      }
  \end{tabular}
 % \centering
\caption{True values, upper bounds, and lower bounds of the three-dimensional estimated states using the differentially private set-based estimator within the context of the CDP setup.}
    \label{fig:cdp-set-based-estimation}%\vspace{-4mm}
    %  \vspace{-0.2cm}

% \graphicspath{ {./Figures/} }
%     \centerline{
% \includegraphics
% [scale=0.38]{Figures/tracking-quadcopter-f}}
%     \caption{Localization of a quadcopter navigating through arbitrary non-linear motion,  {\img{Figures/box}}: rotating object, {\img{Figures/triangle-black}}
%     : sensory nodes, {\color{red}+} : center of the estimated zonotope.}
%     \label{fig:Object-tracking}

\end{figure*}

\begin{figure*}[h]
%\vspace{-0.05cm}
    %\centering
    \begin{tabular}{ p{0.30\textwidth}  p{0.3\textwidth}  p{0.3\textwidth}}
        \resizebox{0.31\textwidth}{!}{
            \begin{subfigure}[h]{0.4\textwidth}
      \centering
        \includegraphics[scale=0.35]{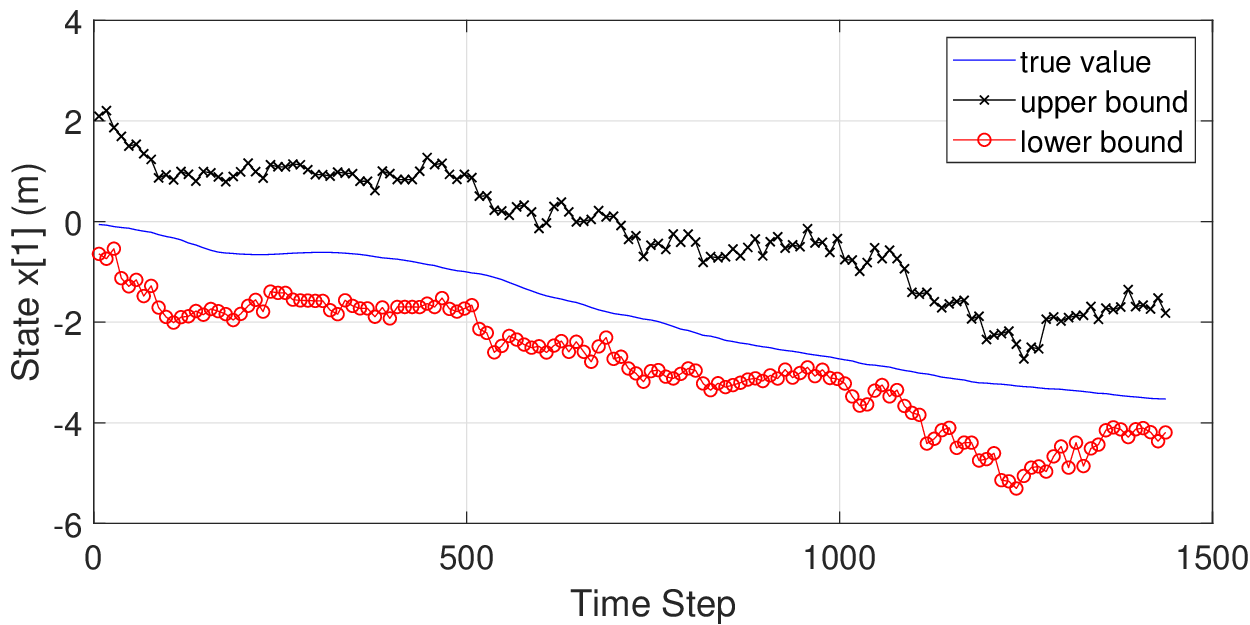}
        \caption{}
        \label{fig:ldp-statex1}
    \end{subfigure}
       } 
   &
   \resizebox{0.31\textwidth}{!}{
            \begin{subfigure}[h]{0.4\textwidth}
      \centering
        \includegraphics[scale=0.35]{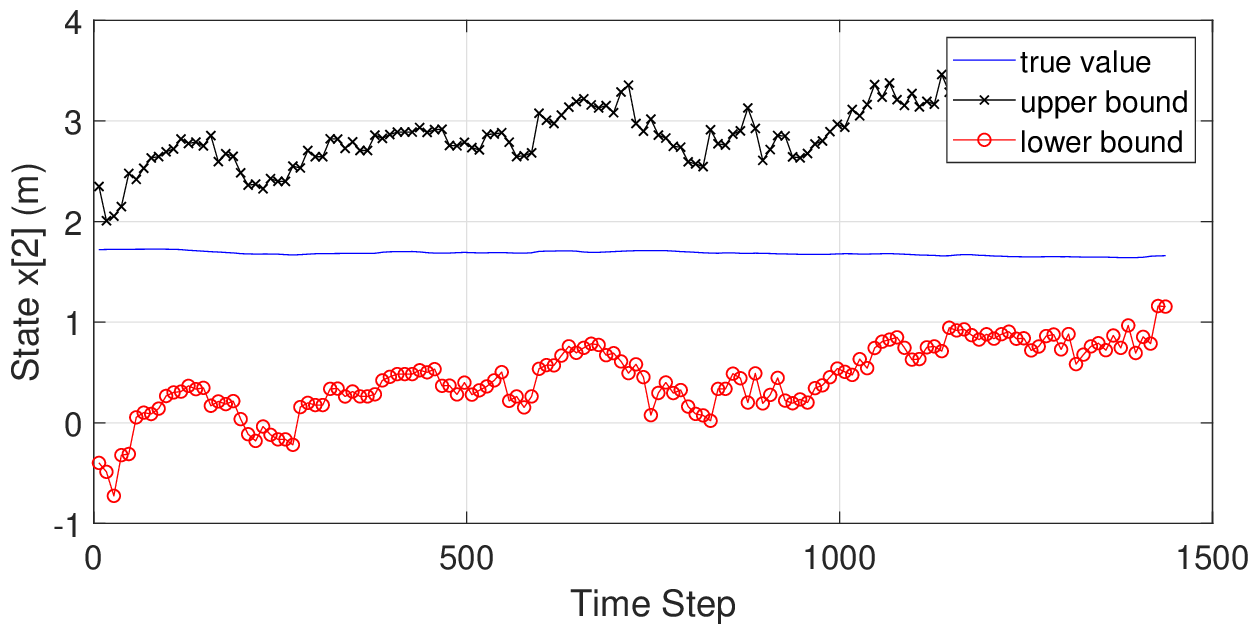}
        \caption{}
        \label{fig:ldp-statex2}
    \end{subfigure}
      }
 &
   \resizebox{0.31\textwidth}{!}{
            \begin{subfigure}[h]{0.4\textwidth}
      \centering
        \includegraphics[scale=0.35]{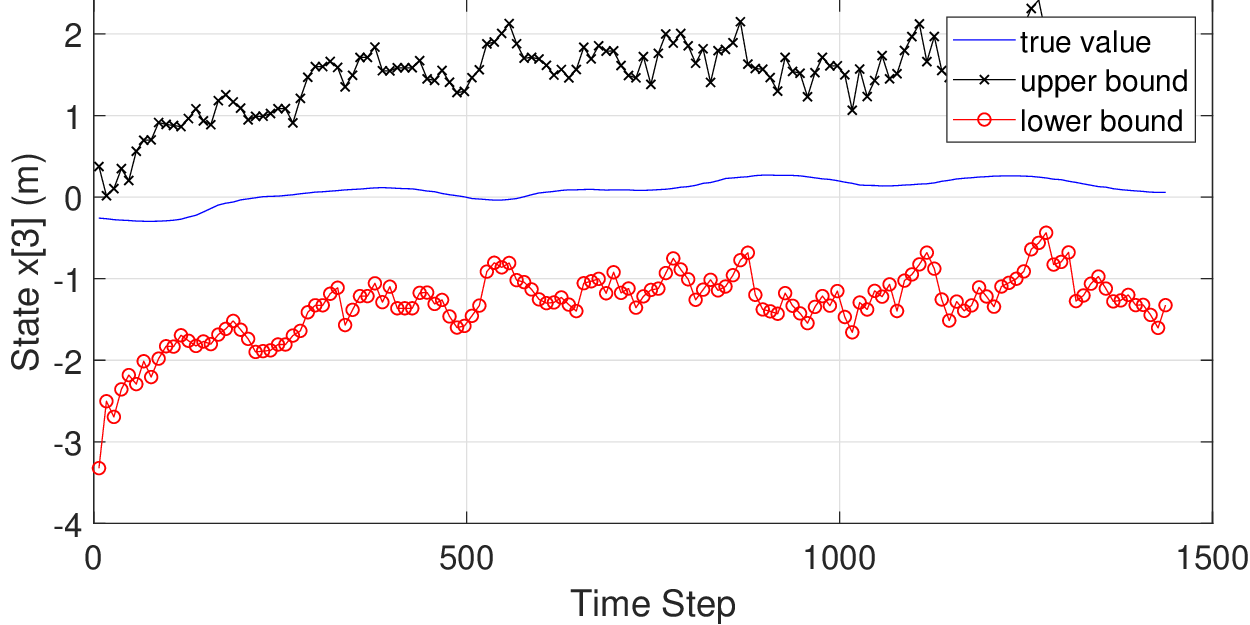}
        \caption{}
        \label{fig:ldp-statex3}
    \end{subfigure}
      }
  \end{tabular}
 % \centering
\caption{True values, upper bounds, and lower bounds of the three-dimensional estimated states using the differentially private set-based estimator within the context of the LDP setup.}
    \label{fig:ldp-set-based-estimation}%\vspace{-4mm}
    %  \vspace{-0.2cm}

% \graphicspath{ {./Figures/} }
%     \centerline{
% \includegraphics
% [scale=0.38]{Figures/tracking-quadcopter-f}}
%     \caption{Localization of a quadcopter navigating through arbitrary non-linear motion,  {\img{Figures/box}}: rotating object, {\img{Figures/triangle-black}}
%     : sensory nodes, {\color{red}+} : center of the estimated zonotope.}
%     \label{fig:Object-tracking}

\end{figure*}

Figure \ref{fig:Object-tracking} represents a random snapshot for the localization of the quadcopter. We notice that the quadcopter is enclosed by the estimated zonotope, which indicates that the state containment is still guaranteed. Also, the center of the estimated zonotope is very close to the quadcopter, which is a good utility indicator. In industrial applications, for a certain privacy budget $\epsilon$, the selection of the optimal noise range $d$, illustrated in Table \ref{table-delta}, should be guided by the acceptable error ranges. Also, Figure \ref{fig:cdp-set-based-estimation} shows the true values, upper bound, and lower bound for each dimension of the estimated state using the proposed differentially private set-based estimator within the context of the CDP setup. 
% Supporting the above, Figure \ref{fig:estimated-location-error-with-epsilon} illustrates that the utility loss, represented by the standard deviation of the error in the estimated location of the rotating object, increases as the privacy budget $\epsilon$ decreases. It also indicates that as the $(\epsilon,\delta)$-ADP optimal noise range widens, there is a noticeable increase in both utility loss and average error in the estimated location of the rotating object.
% %It also shows that utility losses and average errors in estimated locations of the rotating object increase with a wider $(\epsilon,\delta)$-ADP optimal noise range.
% Hence, in industrial applications, for a certain privacy budget $\epsilon$, the selection of the optimal noise range $d$ should be guided by the acceptable error ranges.
% % An industrial application should select the appropriate noise range based on the accepted error ranges.
For comparison, we consider the work in \cite{9147726}, where a differentially private interval estimator deploying truncated Laplace noise is presented. The truncated Laplace noise range, for a given $\epsilon,\delta$, and sensitivity $s$, is determined by 
\begin{equation}
a=\frac{s}{\epsilon}\ln{\Big(1+e^{\epsilon}\frac{1-e^{-\epsilon}}{2\delta}\Big)}.    
\end{equation}

Figure \ref{fig:deltavs-noise-range} indicates that the truncated Laplace noise needed to achieve a certain $\delta$ is wider than the $(\epsilon,\delta)$-ADP truncated optimal noise needed to achieve the same $\delta$. Indeed, for a certain privacy budget $\epsilon$, learning the truncated optimal noise distribution $P(\phi_l)$ using the SGD tool in \cite{DBLP:journals/corr/abs-2107-12957} allows our proposed differentially private set-based estimator (Algorithm \ref{alg:adp-set-estimation-using-zonotope}) to minimize loss of privacy and utility simultaneously. Hence, at a certain $\delta$ value, we find that the truncated Laplace noise causes a higher utility loss than the $(\epsilon,\delta)$-ADP truncated optimal noise.

In particular, we calculate the error in the estimated location as the distance between the estimated zonotope's center and the quadcopter's true location. Then, we compare the utility loss, represented by the second norm of the error in the estimated location of the quadcopter, associated with the $(\epsilon,\delta)$-ADP truncated optimal noise
% (Definition \ref{def:optimal-noise})
and the truncated Laplace noise \cite{9147726} at $\epsilon=0.3$. The results presented in Figure \ref{fig:Estimation-error-vs-noise-acc-vs-optimal-quad-copter-CDP} demonstrate that the utility loss incurred when employing the $(\epsilon,\delta)$-ADP truncated optimal noise is lower than that observed with the truncated Laplace noise at certain privacy level. From the comparative results, we can notice that the proposed differentially private set-based estimator, utilizing $(\epsilon,\delta)$-ADP truncated optimal noise leads to less utility loss compared to employing truncated Laplace noise \cite{9147726} in the context of the CDP setup.
% we compare the utility loss of $(\epsilon,\delta)$-ADP truncated optimal noise 
% (Definition \ref{def:optimal-noise})
% and truncated Laplace noise \cite{9147726} at $\epsilon=0.3$. The results are shown in Figure \ref{fig:Estimation-error-vs-noise-acc-vs-optimal}. 

%For LDP setup,
% The second example is the localization of a quadcopter that moves in arbitrary nonlinear motion within three-dimensional space with dimensions of $10\; \times \; 10 \; \times \; 10 \;m^3$ using a high rate of measurements provided by a set of $8$ anchor nodes distributed across the motion area. The state matrix is as follows: 
% The second example involves 

% The $(\epsilon,\delta)$-ADP truncated optimal noise distribution
% (Definition \ref{def:optimal-noise})
% is regenerated using the same parameters as in the first example except for $d\in[-1,\;1]$. Again, 

%results in less error in the estimated location of the quadcopter than using the truncated Laplace noise distribution. 

\begin{figure}[h]
\graphicspath{ {./Figures/} }
    \centerline{
\includegraphics
[scale=0.38]{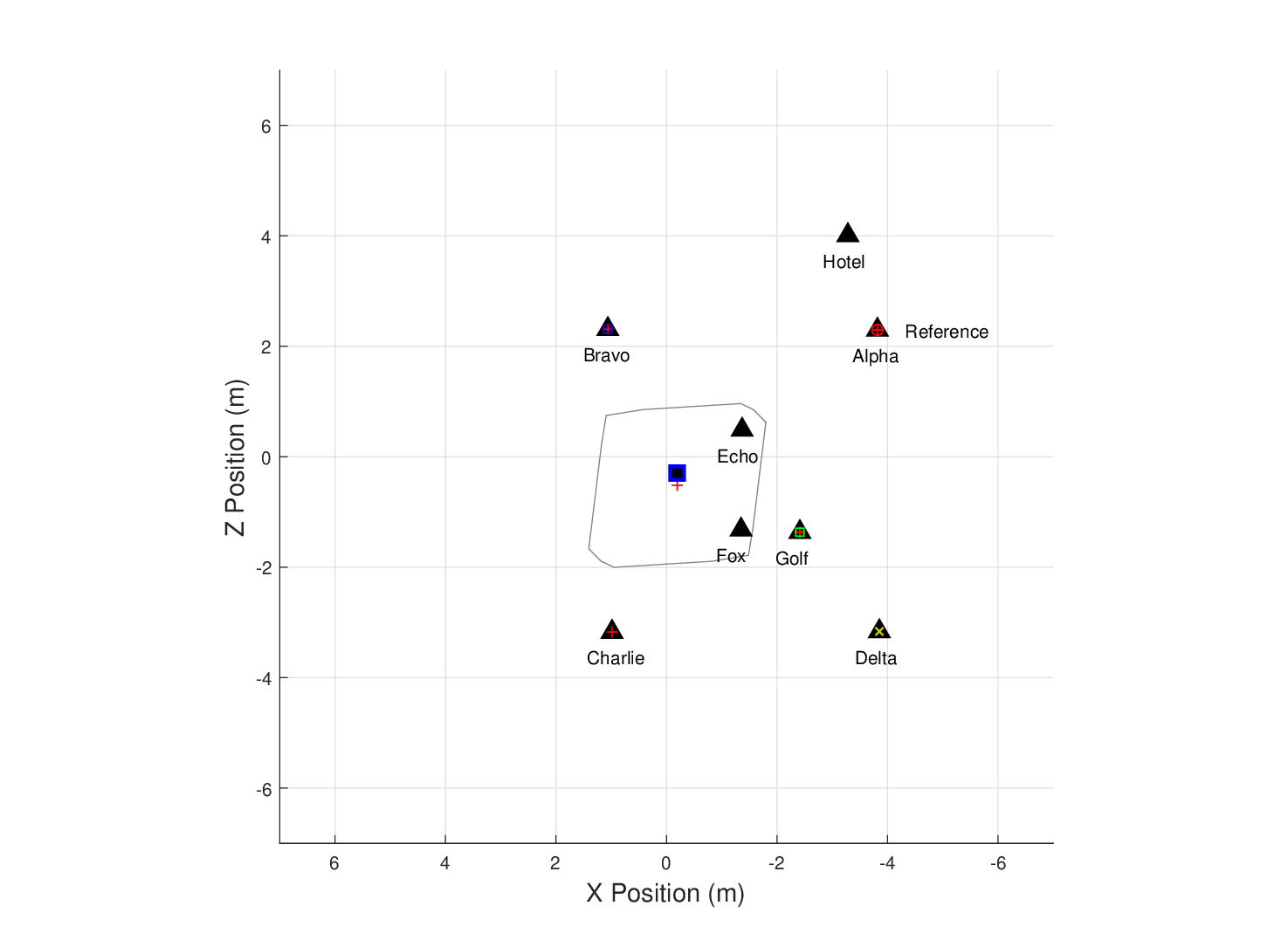}}
    \caption{Localization of a quadcopter navigating through arbitrary non-linear motion,  {\img{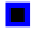}}: the quadcopter, {\img{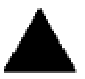}}
    : anchor nodes, {\color{red}+} : center of the estimated zonotope.}
    \label{fig:Object-tracking}
\end{figure}

\begin{table*}[htbp]\centering 
%\ra{1.3}
\begin{small}
\begin{tabular}{@{}lrrrrrrrrrrrr@{}}\toprule
{$\epsilon$\;/\;Noise Range $d(m)$} & {3} & {5} & {7} & {9} & {11} & {13} & {15} \\ \midrule
{0.1} & 0.1502\ & 0.0811\ & 0.0518\ & 0.0360\ & 0.0262\ & 0.0197\ & 0.0151\ \\ 
{0.3} & 0.1198\ & 0.0503\ & 0.0244\ & 0.0126\ & 0.0067\ &0.0036\ & 0.0020\ \\ 
{0.5} & 0.0931\ & 0.0290\ & 0.0101\ & 0.0036\ & 0.0013\ & 0.0005\ & 0.0002\ \\ %\hdashline
{0.7} & 0.0707\ & 0.0158\ & 0.0038\ & 0.0009\ & 0.0002\ & $5.64e^{-5}$\ & $1.39e^{-5}$\ \\ %\hdashline
\bottomrule
\end{tabular}
\end{small}
\caption{Optimal $\delta$ values corresponding to different ranges $d(m)$ of $(\epsilon,\delta)$-ADP optimal noise at $\epsilon={0.1,0.3,0.5,0.7}$.}\label{table-delta}
\end{table*}

\subsection{The LDP Model}\label{Evaluation-for-The-LDP-Setup}

% We conduct the evaluation using the localization of a quadcopter navigating arbitrary non-linear motion within three-dimensional space, with dimensions of $10\; \times \; 10 \; \times \; 10 \;m^3$. The quadcopter's localization is achieved through a high rate of measurements provided by a set of $8$ anchor nodes distributed across the motion area. 

In this subsection, we evaluate the proposed differentially private set-based estimator within the context of the LDP setup. It deserves mentioning that the sensor manager is not present to aggregate the measurements of the sensors since it is treated as an untrusted entity in the LDP setup. Instead, each sensor locally perturbs its own measurement with the privacy-preserving noise. As in the CDP setup of this example, the $(\epsilon,\delta)$-ADP truncated optimal noise distribution is regenerated for the LDP setup. Figure \ref{fig:ldp-set-based-estimation} shows the true values, upper bound, and lower bound for each dimension of the estimated state using the proposed differentially private set-based estimator within the context of the LDP setup. In the context of the LDP setup, we again compare the utility loss of $(\epsilon,\delta)$-ADP truncated optimal noise 
% (Definition \ref{def:optimal-noise}) 
and truncated Laplace noise \cite{9147726} at $\epsilon=0.3$. 
% The results are shown in Figure \ref{fig:Estimation-error-vs-noise-acc-vs-optimal-quad-copter-LDP}.
The comparative results in Figure \ref{fig:Estimation-error-vs-noise-acc-vs-optimal-quad-copter-LDP} demonstrate that in the proposed differentially private set-based estimator, utilizing $(\epsilon,\delta)$-ADP truncated optimal noise leads to less utility loss compared to employing truncated Laplace noise \cite{9147726} in the context of the LDP setup.
% From the comparative results in Figure \ref{fig:Estimation-error-vs-noise-acc-vs-optimal-quad-copter-LDP}, we can notice that the proposed $(\epsilon,\delta)$-ADP set-based estimator, utilizing the $(\epsilon,\delta)$-ADP truncated optimal noise, also outperforms the differentially private interval estimator, employing truncated Laplace noise \cite{9147726}, in the context of LDP setup by achieving less utility loss.
%resulting in less utility loss using the LDP setup.

% Additionally, our proposed $(\epsilon,\delta)$-ADP set-based estimator (Algorithm \ref{alg:adp-set-estimation-using-zonotope}) leverages zonotopes for set representation. The zonotopes are less conservative set representation that allow efficient computation of linear maps and Minkowski sum, which are essential operations at the set-based estimation process. 
% This in turn gives the proposed $(\epsilon,\delta)$-ADP set-based estimator a computational advantage.
Furthermore, our proposed differentially private set-based estimator (Algorithm \ref{alg:adp-set-estimation-using-zonotope}) utilizes zonotopes for set representation. Zonotopes offer a less conservative set representation that enables efficient computation of linear maps and Minkowski sums, both essential operations in the set-based estimation process.
Consequently, this confers a computational advantage to the proposed differentially private set-based estimator.

% \begin{figure}[t]
% \graphicspath{ {./Figures/} }
%     \centerline{
% \includegraphics
% [scale=0.56]{Figures/Estimation-error-vs-noise-range2}}
%     \caption{Effect of $(\epsilon,\delta)$-ADP optimal noise range on the average error in the estimated location of the rotating object. }
%     \label{fig:estimated-location-error-with-epsilon}
% \end{figure}

% \begin{figure}[t]
% \graphicspath{ {./Figures/} }
% \centerline{\includegraphics[scale=0.56]{Estimation-error-vs-noise-acc-vs-optimal2}}
% \caption{
% % Comparing the effect of $(\epsilon,\delta)$-ADP truncated optimal noise and truncated Laplace noise on the error at the estimated location of the rotating object at $\epsilon=0.3$.
% Comparison of the impact of $(\epsilon, \delta)$-ADP truncated optimal noise and truncated Laplace noise on the error in the estimated location of the rotating object, with $\epsilon = 0.3$.}
% \label{fig:Estimation-error-vs-noise-acc-vs-optimal}
% \end{figure}

%% file: Sections/6-con.tex
\section{Conclusion} \label{sec:conc}
We have proposed a differentially private set-based estimator that performs set-based estimation in a privacy-preserving manner so that an adversary cannot learn the actual values of measurements from state estimates. 
% It deploys an additive mechanism with truncated optimal noise that holds the measurements private while minimizing the utility loss.
% The proposed $(\epsilon,\delta)$-ADP set-based estimator ensures the privacy of sets containing sensors' measurements throughout the estimation process with minimal utility loss. The evaluation results for our proposed estimator illustrate that a wider range of the truncated Laplace noise than of the $(\epsilon,\delta)$-ADP truncated optimal noise is needed to achieve a specific value of $\delta$. Hence, the truncated optimal noise has a lower utility loss than the truncated Laplace noise with the same $\epsilon$ and $\delta$ values for both the central and the local differential privacy models. Also, using zonotopes for set representation in our proposed $(\epsilon,\delta)$-ADP set-based estimator gives it a computational advantage.
The proposed differentially private set-based estimator ensures the privacy of the measurement sets containing sensitive sensors' measurements throughout the estimation process with minimal utility loss. Evaluation results demonstrate that a wider range of truncated Laplace noise is needed compared to $(\epsilon,\delta)$-ADP truncated optimal noise to achieve a specific value of $\delta$. Consequently, the truncated optimal noise incurs lower utility loss than truncated Laplace noise with equivalent $\epsilon$ and $\delta$ values under both central and local differential privacy models. Additionally, using zonotopes for set representation in our proposed differentially private set-based estimator confers a computational advantage.

\begin{figure} [t]
\graphicspath{ {./Figures/} }
\centerline{\includegraphics[scale=0.56]{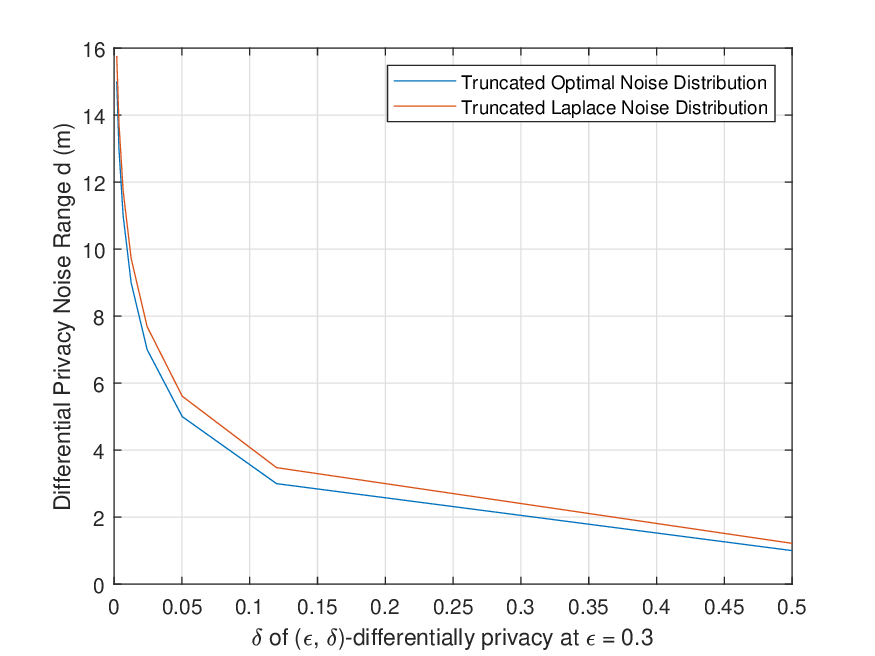}}
\caption{Comparison of $\delta$ values with needed range of both $(\epsilon,\delta)$-ADP truncated optimal noise and truncated Laplace noise. }
\label{fig:deltavs-noise-range}
\end{figure}

\begin{figure}[t]
\graphicspath{ {./Figures/} }
\centerline{\includegraphics[scale=0.56]{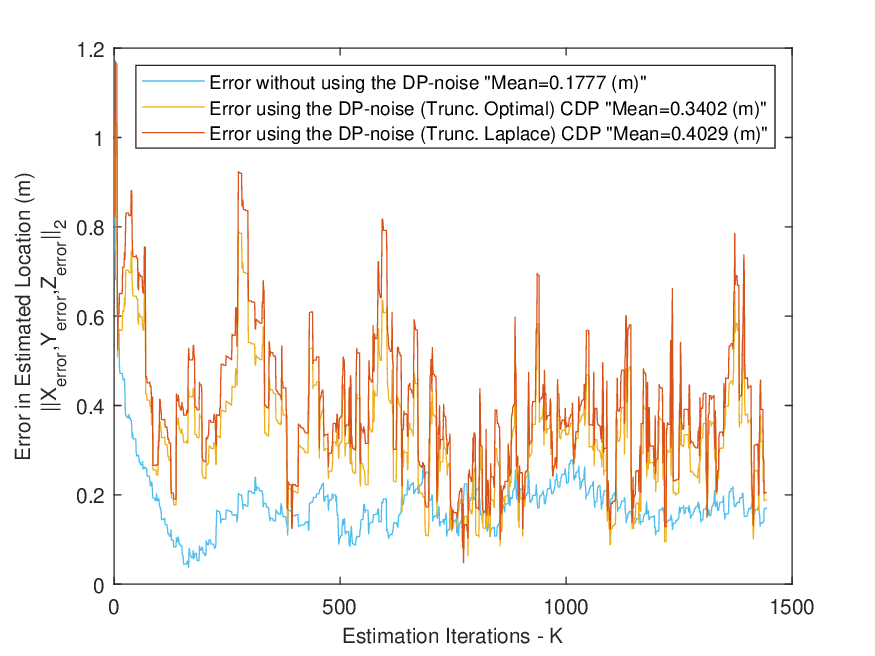}}
\caption{
% Comparing the effect of $(\epsilon,\delta)$-ADP truncated optimal noise and truncated Laplace noise on the error at the estimated location of the quad-copter at $\epsilon=0.3$ for the CDP setup.
Comparison of the impact of $(\epsilon,\delta)$-ADP truncated optimal noise and truncated Laplace noise on the error in the estimated location of the quadcopter at $\epsilon=0.3$, within the CDP setup.}
\label{fig:Estimation-error-vs-noise-acc-vs-optimal-quad-copter-CDP}
\end{figure}

\begin{figure}[t]
\graphicspath{ {./Figures/} }
\centerline{\includegraphics[scale=0.56]{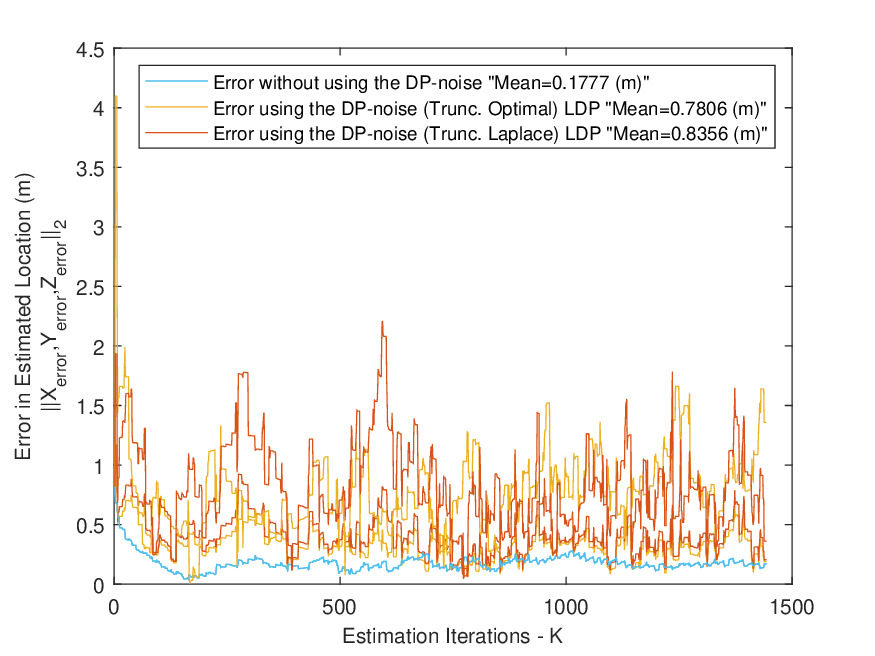}}
\caption{Comparison of the impact of $(\epsilon,\delta)$-ADP truncated optimal noise and truncated Laplace noise on the error in the estimated location of the quadcopter at $\epsilon=0.3$, within the LDP setup.}
\label{fig:Estimation-error-vs-noise-acc-vs-optimal-quad-copter-LDP}
\end{figure}